\documentclass[runningheads]{llncs}
\newcommand{\hide}[1]{}  
\usepackage{textcomp}
\usepackage{bussproofs} 
\usepackage{multicol}
\usepackage{color}  
\usepackage{algpseudocode}
\usepackage{times} 
\usepackage{tikz}   
\usetikzlibrary{fit,decorations.shapes,decorations.pathmorphing,shapes,arrows}
\usepackage{verbatim}
\usepackage{amssymb}  
\usepackage{mathbbol} 
\usepackage{cmll}
\usepackage[lofdepth,lotdepth]{subfig}
\usepackage{amsfonts}
\usepackage{amsmath} 
\usepackage{multirow}
\usepackage{wrapfig}
\usepackage{cite}
\newcommand{\formulasize}{\ensuremath{\code{f}\_\code{size}}}
\newcommand{\abot}{\ensuremath{\bot}} 
\newcommand{\mtop}{\ensuremath{^*\!\!\top}}
\newcommand{\lollipop}{\ensuremath{-\!\circ}}
\newcommand{\dpreceq}{\ensuremath{\hat{\preceq}}}
\newcommand{\IL}{\code{IL}}  
\newcommand{\MILL}{\code{MILL}}

\newcommand{\Lwedge}
{\AxiomC{$\Gamma(F; G) \vdash H$}  
\RightLabel{$\wedge L$} 
\UnaryInfC{$\Gamma(F \wedge G) \vdash H$} 
\DisplayProof
} 
\newcommand{\Rwedge} 
{\AxiomC{$\Gamma \vdash F$} 
\AxiomC{$\Gamma \vdash G$} 
\RightLabel{$\wedge R$} 
\BinaryInfC{$\Gamma \vdash F \wedge G$} 
\DisplayProof 
} 
\newcommand{\Lvee} 
{\AxiomC{$\Gamma(F) \vdash H$} 
\AxiomC{$\Gamma(G) \vdash H$} 
\RightLabel{$\vee L$} 
\BinaryInfC{$\Gamma(F \vee G) \vdash H$} 
\DisplayProof 
} 
\newcommand{\Rvee} 
{\AxiomC{$\Gamma \vdash F_i$}
\RightLabel{$\vee R$} 
\UnaryInfC{$\Gamma \vdash F_1 \vee F_2$} 
\DisplayProof} 
\newcommand{\LrightarrowTH}
{\AxiomC{$\mathbb{E}(\widetilde{\Gamma_1}; F {\supset} G) \vdash F$} 
    \AxiomC{$\Gamma(G; \mathbb{E}(\widetilde{\Gamma_1}; F {\supset} G)) \vdash H$} 
\RightLabel{$\supset L$} 
\BinaryInfC{$\Gamma(\mathbb{E}(\widetilde{\Gamma_1}; F {\supset} G)) \vdash H$} 
\DisplayProof}
\newcommand{\RrightarrowTH} 
{\AxiomC{$\Gamma; F \vdash G$} 
\RightLabel{$\supset R$} 
\UnaryInfC{$\Gamma \vdash F {\supset} G$} 
\DisplayProof}
\newcommand{\Lstar} 
{\AxiomC{$\Gamma(F, G) \vdash H$} 
\RightLabel{$* L$} 
\UnaryInfC{$\Gamma(F * G) \vdash H$} 
\DisplayProof}

\newcommand{\LstararrowTHO}
{
\AxiomC{$Re_1 \vdash F$} 
\AxiomC{$\Gamma((Re_2, G); (\Gamma', \mathbb{E}(\Gamma_1;
F \text{\wand} G))) \vdash H$} 
\RightLabel{$\text{\wand} L_1$} 
\BinaryInfC{$\Gamma(\Gamma', \mathbb{E}(\Gamma_1; F {\text{\wand}} G)) \vdash H$} 
\DisplayProof} 
\newcommand{\LstararrowTHT} 
{
\AxiomC{$\Gamma' \vdash F$} 
\AxiomC{$\Gamma(G; (\Gamma', \mathbb{E}(\Gamma_1; F {\text{\wand}} G))) \vdash H$} 
\RightLabel{$\text{\wand} L_2$} 
\BinaryInfC{$\Gamma(\Gamma', \mathbb{E}(\Gamma_1; F {\text{\wand}} G)) \vdash H$} 
\DisplayProof}
\newcommand{\LstararrowTHTH}
{
\AxiomC{$\O_m \vdash F$} 
\AxiomC{$\Gamma((\Gamma', G); (\Gamma', 
\mathbb{E}(\Gamma_1; F \text{\wand} G))) \vdash H$} 
\RightLabel{$\text{\wand} L_3$} 
\BinaryInfC{$\Gamma(\Gamma', \mathbb{E}(\Gamma_1; F {\text{\wand}} G)) \vdash H$} 
\DisplayProof} 
\newcommand{\LstararrowTHF} 
{ 
\AxiomC{$\O_m \vdash F$} 
\AxiomC{$\Gamma(G; F {\text{\wand}} G) \vdash H$} 
\RightLabel{$\text{\wand} L_4$} 
\BinaryInfC{$\Gamma(F {\text{\wand}} G) \vdash H$} 
\DisplayProof}  
\newcommand{\Rstararrow} 
{ 
\AxiomC{$\Gamma, F \vdash G$} 
\RightLabel{$\text{\wand} R$} 
\UnaryInfC{$\Gamma \vdash F {\text{\wand}} G$} 
\DisplayProof
}
\newcommand{\IDTH} 
{ 
\AxiomC{}
\RightLabel{$id$} 
\UnaryInfC{$\mathbb{E}(\widetilde{\Gamma}; p) \vdash p$} 
\DisplayProof}

\newcommand{\Rtop} 
{ 
\AxiomC{}
\RightLabel{$\top R$} 
\UnaryInfC{$\Gamma \vdash \top$} 
\DisplayProof}
 
\newcommand{\Lbot} 
{
\AxiomC{}
\RightLabel{$\abot L$} 
\UnaryInfC{$\Gamma(\abot) \vdash F$} 
\DisplayProof}
\newcommand{\Cut}{\text{\code{Cut}}}

\newcommand{\repeatnormalise}{\ensuremath{\code{repeat}\_\code{norm}}}

\newcommand{\norm}{\code{norm}}

\newcommand{\derivationlength}{\ensuremath{\code{der}\_\code{len}}}

\newcommand{\code}[1]{{\ensuremath{\tt #1}}} 

\newcommand{\GOI}{\code{G1i}}
\newcommand{\GTI}{\code{G3i}}
\newcommand{\GFI}{\code{G4i}}

\newcommand{\BI}{\code{BI}} 

\newcommand{\LBI}{\code{LBI}}
\newcommand{\LBIT}{\code{LBI2}} 
\newcommand{\LBITH}{\ensuremath{\alpha\code{LBI}}}

\newcommand{\derivationdepth}{\code{der\_depth}}

\newcommand{\wand}{$-\!\!*$}

\newtheorem{observation}{Observation}
\newcommand{\LBIN}{\ensuremath{\LBI}\code{Z}}
\newcommand\realtenfootnotesize{%
   \@setfontsize\footnotesize\@viiipt{9.5}%
} 

\hide{
\documentclass[a4paper]{IEEEtran}
\usepackage{bussproofs}
\usepackage{color}
\usepackage{times}
\usepackage{tikz} 
\usetikzlibrary{decorations.shapes,decorations.pathmorphing,shapes,arrows,fit}
\usepackage{verbatim} 
\usepackage{wrapfig}
\usepackage{amssymb}
\usepackage{amsfonts}
\usepackage{amsmath} 
\usepackage{cite}
\newcommand{\hide}[1]{}
  
\newcommand{\GOI}{\textbf{G1i}} 
  
\newcommand{\GTI}{\textbf{G3i}}
\newcommand{\GFI}{\textbf{G4i}}

\newcommand{\BI}{\textbf{BI}}
\newcommand{\LBI}{\textbf{LBI}}
\newcommand{\LBIT}{\textbf{LBI2}} 

\newcommand{\LBITH}{\textbf{LBI3}}

\newcommand{\wand}{\ensuremath{-\!\!*}} 

\newtheorem{observation}{Observation}

\newtheorem{example}{Example}
\newtheorem{lemma}{Lemma}
\newtheorem{theorem}{Theorem}
\newtheorem{proposition}{Proposition}
\newtheorem{definition}{Definition}

} 

\begin{document} 
\title{Structural Interactions and Absorption of Structural 
Rules in BI 
Sequent Calculus}
\author{Ryuta Arisaka}
\institute{INRIA Saclay--\^Ile-de-France, 
        Campus de l'\'Ecole Polytechnique} 
\maketitle
\begin{abstract}        
    Development of a contraction-free {\BI} sequent 
    calculus, be it in the 
    sense of {\GTI} or \GFI, has not been successful in literature. 
   We address the open problem by 
   presenting such a sequent system. 
   In fact our calculus involves no structural 
   rules. 
  \hide{
  Two key concepts for the solutions are 
  (1) essence of structures that recognises and bundles within 
  a sequent 
  relevant structures required for a logical inference rule, and 
  (2) deep absorption of 
  {\LBI} weakening. The former is a notational invention
  addressing 
  interactions between {\LBI} logical inference rules, 
  weakening and the structural equivalence around the multiplicative 
  unit, whereas the latter is a mechanism 
  that gives rise to a critical 
  observation of incremental weakening isolating
  the effect of contraction from that of weakening. 
  Cut admissibility in {\LBITH} is then proved. 
  The sequent calculus is further stripped off redundant artifacts;
  structural units dispelled altogether, into 
  \LBIN. The transition challenges 
  the widely followed wisdom of 
  coherent equivalence in {\BI} proof systems. 
\hide{  From the perspective of proof searches, the presence of 
  contraction and {\BI} structural equivalences are 
  bottlenecks in earlier {\BI} sequent calculi such as 
  \LBI. A contraction-free {\BI} sequent calculus has remained 
  in obscurity due to the difficulty of analysing interactions 
  among inference rules, not only those between {\LBI} contraction and 
  its logical inference rules, but also those between {\LBI} contraction 
  and the other 
  {\LBI} structural rules. }   Moving on, we 
  derive a purely syntactic decidability result for a subset of 
  {\BI} without the multiplicative unit and the multiplicative 
  implication, based on \LBIN. 
  Towards the goal, we incorporate a well-known implicit 
  contraction elimination technique by Dyckhoff and show 
  a successful extension beyond {\IL} logical connectives. 
  }
\end{abstract}
\section{Introduction}      
Propositional {\BI} \cite{OHearnP99} is a conservative extension 
of propositional intuitionistic logic {\IL}
and propositional multiplicative fragment of 
intuitionistic linear logic {\MILL} 
(\emph{Cf.} \cite{DBLP:journals/tcs/Girard87} for linear logic). 
It is conservative 
in the sense that all the theorems of {\IL} and {\MILL} 
are a theorem of \BI. But the extension is not the least 
conservative. That is, there are expressions of {\BI} 
that are not expressible in {\IL} or {\MILL} \cite{OHearnP99}. 
They shape logical characteristics unique to \BI, which 
must be studied. Structural interactions in sequent calculus 
(interactions between logical rules and structural 
rules)
is one of them for which the details matter. Earlier works
\cite{journals/mscs/GalmicheMJP05,Harland03,DonnellyGKMP04,Brotherston10-4,OHearn03} on {\BI} appear to 
suggest that the study is non-trivial, however. 
In this work  
we solve an open problem of absorption of 
structural rules,  which is of 
theoretical interest having a foundational implication
to automated reasoning. 
Techniques 
considered here should be of interest to proof-theoretical studies of 
other non-classical logics.  
\subsection{Logic \BI} 
{\BI} has a proof-theoretical origin.  
A proof system was defined \cite{OHearnP99}, 
followed by semantics
\cite{Pym02,journals/mscs/GalmicheMJP05}. 
To speak of the language of {\BI} first, if 
we denote propositional variables by {\small $\mathcal{P}$},
signatures of {\IL} by {\small $\{\top_0, \abot_0, \wedge_2, 
        \vee_2, \supset_2\}$} and 
those of {\MILL} by {\small $\{\mtop_0, *_2, \text{\wand}_2\}$}\footnote{The sub-scripts denote the arity.}
where {\small $\mtop$} is the multiplicative top element 
{\small $\textbf{1}$}, {\small $*$} is linear `times' 
{\small $\otimes$} and 
{\small $\text{\wand}$} is linear implication 
{\small $\lollipop$} \cite{DBLP:journals/tcs/Girard87}, 
then it comprises all the expressions that are constructable 
from {\small $(\mathcal{P}, 
    \{\top_0, \abot_0, \mtop_0, \wedge_2, \vee_2, \supset_2, 
        *_2, \text{\wand}_2\})$}. Let us suppose 
two arbitrary expressions (formulas) {\small $F$} and {\small $G$} 
in the language. Then like in \IL, we can construct 
{\small $F \wedge G, F \vee G, F \supset G$}; and, like 
in \MILL, we can construct {\small $F * G, F \text{\wand} G$}.  
The two types are actively distinguished in {\BI} proof systems by
two distinct structural connectives. The below examples are 
given in 
\cite{OHearnP99}. 
\begin{center} 
  {\small
  \AxiomC{$\Gamma; F \vdash G$} 
  \RightLabel{$\supset R$} 
  \UnaryInfC{$\Gamma \vdash F {\supset} G$} 
  \DisplayProof 
  \indent 
  \AxiomC{$\Gamma, F \vdash G$} 
  \RightLabel{$\text{\wand} R$} 
  \UnaryInfC{$\Gamma \vdash F {\text{\wand}} G$}
  \DisplayProof 
  }
\end{center}    
{\small $\Gamma$} denotes a structure.\footnote{Those proof-theoretical
     terms 
    are assumed familiar. They are found for example 
    in \cite{351148}. But formal definitions that we will need
    for technical discussions 
    will be found in the next section.}
 Note the use of 
two structural connectives ``;'' and ``,'' for 
a structural distinction. 
If there were only ``,'', 
both {\small $\supset R$} and {\small ${\text{\wand}} R$} 
could apply on {\small $\Gamma, F \vdash G$}. The contextual differentiation
is a simple way to isolate the two implications. Following 
the convention of linear logic, the {\IL} structures that 
``;'' form are termed additive; and the {\MILL} 
structures multiplicative, similarly. One axiom: 
{\small $F = F \wedge \top = F * \mtop$}, connects 
the two types. But ``;'' and ``,'' do not distribute 
over one another. So in general 
a {\BI} structure    is a nesting of additive structures 
    {\small $\Gamma_1; \Gamma_2$} and multiplicative 
    structures {\small $\Gamma_1, \Gamma_2$}.  
    In the first {\BI} sequent calculus {\LBI} \cite{Pym02}, 
    we have the following structural rules as expected: 
\begin{center} 
  {\small
  \AxiomC{$\Gamma(\Gamma_1; \Gamma_1) \vdash F$} 
  \RightLabel{Contraction} 
  \UnaryInfC{$\Gamma(\Gamma_1) \vdash F$} 
  \DisplayProof 
  \indent 
  \AxiomC{$\Gamma(\Gamma_1) \vdash F$} 
  \RightLabel{Weakening} 
  \UnaryInfC{$\Gamma(\Gamma_1; \Gamma_1) \vdash F$} 
  \DisplayProof
  }
\end{center}  
where {\small $\Gamma(...)$} abstracts any other structures 
surrounding the focused ones in the sequents. We will formally define 
the notation later. 
\subsection{Research problems and contributions}  
The formulation of {\BI} is intuitive, as we just saw. 
But that {\BI} is not the least conservative extension 
of {\IL} and {\MILL} means that {\IL} and {\MILL} interact 
in parts of \BI. Structurally we have an interesting phenomenon. 
When we consider instances of the contraction rule 
as were stated earlier, we find that there are several of them, 
including ones below. 
\begin{center}
  {\small
  \AxiomC{$\Gamma((F; F), G) \vdash H$} 
  \RightLabel{$Ctr_1$} 
  \UnaryInfC{$\Gamma(F, G) \vdash H$} 
  \DisplayProof 
  \indent 
  \AxiomC{$\Gamma(F, (G;G)) \vdash H$} 
  \RightLabel{$Ctr_2$} 
  \UnaryInfC{$\Gamma(F, G) \vdash H$} 
  \DisplayProof 
  \AxiomC{$\Gamma((F, G); (F, G)) \vdash H$}  
  \RightLabel{$Ctr_3$}
  \UnaryInfC{$\Gamma(F, G) \vdash H$} 
  \DisplayProof 
  }
\end{center}        
The first two are simply {\GOI} \cite{351148} contractions. 
The last is not, since what is duplicating bottom-up  
is a structure. And it poses some proof-theoretical problem: 
if it is not admissible\footnote{An inference rule in  sequent calculus
is admissible when any sequent which is derivable in the calculus 
is derivable without the particular rule.}
in \LBI, we cannot impose any general restriction 
on the size of what may duplicate bottom-up, and contraction analysis
becomes non-trivial. As we are to state in due course, 
indeed structural contraction is not admissible in \LBI. For a successful 
contraction absorption, we need to identify what 
in {\LBI} require the general contraction. \\
\indent Two issues
stand in the way of a successful {\LBI} contraction analysis, however.
The first is the structural equivalences
{\small $\Gamma, \O_m = \Gamma = \Gamma; \O_a$} (where {\small $\O_a$} 
denotes the additive nullary structural connective corresponding 
to {\small $\top$} and {\small $\O_m$} 
the multiplicative 
nullary structural connective corresponding to 
{\small $\mtop$}) which are by nature bidirectional: 
{\small 
\begin{center}
  \AxiomC{$\Gamma \vdash F$}  
  \UnaryInfC{$\Gamma; \O_a \vdash F$} 
  \DisplayProof 
  {\ }
  \AxiomC{$\Gamma; \O_a \vdash F$} 
  \UnaryInfC{$\Gamma \vdash F$} 
  \DisplayProof 
  {\ } 
  \AxiomC{$\Gamma \vdash F$} 
  \UnaryInfC{$\Gamma, \O_m \vdash F$} 
  \DisplayProof
  {\ }
  \AxiomC{$\Gamma, \O_m \vdash F$} 
  \UnaryInfC{$\Gamma \vdash F$} 
  \DisplayProof
\end{center} 
} 
\noindent Apart from being an obvious source of non-termination, 
it obscures
the core mechanism of 
structural interactions by seemingly implying
a free transformation
of an additive structure into a multiplicative one and vice versa.  
The second is the difficulty of isolating the effect of contraction 
from that of weakening, as a work by Donnelly \emph{et al} 
\cite{DonnellyGKMP04}
experienced (where
contraction is absorbed into weakening as well as into logical 
rules). It is also not so 
straightforward to
know whether, first of all, either weakening or contraction is immune 
to the effect of the structural equivalences. 
As the result of the technical complications, contraction-free {\BI} sequent calculi, be the 
contraction-freeness
in the sense of {\GTI} or of {\GFI} \cite{351148,Dyckhoff92},  
have remained in obscurity. \\
\indent The current status of the knowledge 
of structural interactions within {\BI} proof systems is   
not very satisfactory. 
From the perspective of 
theorem proving for example, the presence of the bidirectional 
rules and contraction as explicit structural rules
in {\LBI} means that it is difficult to actually prove that
an invalid {\BI} proposition is underivable 
within the calculus. This is because {\LBI} by itself does not provide termination 
conditions apart when a (backward) derivation actually terminates:   
the only case in which no more backward derivation on 
a {\LBI} sequent is possible is 
when the sequent is empty; the only case in which 
it is empty is when it is the premise of an axiom.\\
\indent We solve the open problem of contraction absorption, 
but even better, of absorbing all the structural rules. We also 
eliminate nullary structural connectives. 
The objective of this work is to solve 
the mentioned long unsolved open problem in proof theory. 
We do not even require an explicit semantics 
introduction. Therefore technical dependency on earlier 
works is pretty small. Only the knowledge 
of {\LBI} \cite{Pym02} is required. 
\subsection{Structure of the remaining sections} 
In Section 2 we present technical preliminaries 
of {\BI} proof theory. 
In Section 3 we introduce our {\BI} calculus {\LBIN} with 
no structural rules. In Section 4 we 
show its main properties including  admissibility  
of structural rules and equivalence to \LBI. We also 
show {\Cut} admissibility in [{\LBIN} + \Cut]. 
Section 5 concludes. 
\section{BI Proof Theory - Preliminaries}   
We assume the availability of the following meta-logical notations. ``If and only if'' 
is abbreviated by ``iff''. 
\begin{definition}[Meta-connectives]
  We denote logical conjunction (``and'') by {\small $\wedge^{\dagger}$}, 
  logical disjunction (``or'') by {\small $\vee^{\dagger}$}, 
  material implication (``implies'') by {\small $\rightarrow^{\dagger}$}, and 
  equivalence by {\small $\leftrightarrow^{\dagger}$}. 
  These follow the semantics of standard classical logic's.  
\end{definition}
We denote propositional variables by {\small $\mathcal{P}$} 
and refer to an element of {\small $\mathcal{P}$} by 
{\small $p$} or {\small $q$} with or without a sub-script.  
\begin{figure*}[!t]
  \renewcommand{\arraystretch}{1.3}
  \centering
  \scalebox{0.85}{ 
  \begin{tabular}{cccc}  
    \AxiomC{}
    \RightLabel{id}
    \UnaryInfC{$F \vdash F$}
    \DisplayProof
    &   
    \AxiomC{$\Gamma_1 \vdash G$}
    \AxiomC{$\Gamma(G) \vdash H$}
    \RightLabel{Cut}
    \BinaryInfC{$\Gamma(\Gamma_1) \vdash H$}
    \DisplayProof
    &
    \AxiomC{}
    \RightLabel{$\bot L$}
    \UnaryInfC{$\Gamma(\bot) \vdash H$}
    \DisplayProof 
    \\&&\\
    \AxiomC{}
    \RightLabel{$\top R$}
    \UnaryInfC{$\Gamma \vdash \top$}
    \DisplayProof
    &
    \AxiomC{} 
    \RightLabel{$\mtop R$}
    \UnaryInfC{$\mtop \vdash {\mtop}$}
    \DisplayProof
    &
    \AxiomC{$\Gamma(F; G) \vdash H$}
    \RightLabel{$\wedge L$}
    \UnaryInfC{$\Gamma (F \wedge G) \vdash H$}
    \DisplayProof
    \\&&\\
    \AxiomC{$\Gamma(F) \vdash H$}
    \AxiomC{$\Gamma(G) \vdash H$}
    \RightLabel{$\vee L$}
    \BinaryInfC{$\Gamma(F \vee G) \vdash H$}
    \DisplayProof
    & 
    \multicolumn{2}{c}{ 
    \AxiomC{$\Gamma_1 \vdash F$}
    \AxiomC{$\Gamma(\Gamma_1; G) \vdash H$}
    \RightLabel{$\supset L$}
    \BinaryInfC{$\Gamma(\Gamma_1; F \supset G) \vdash H$}
    \DisplayProof}
\\&&\\
    \AxiomC{$\Gamma(F, G) \vdash H$}
    \RightLabel{$* L$}
    \UnaryInfC{$\Gamma(F * G) \vdash H$}
    \DisplayProof
    &
    \AxiomC{$\Gamma_1 \vdash F$}
    \AxiomC{$\Gamma(G) \vdash H$}
    \RightLabel{$\text{\wand} L$}
    \BinaryInfC{$\Gamma(\Gamma_1, F \text{\wand} G) \vdash H$}
    \DisplayProof
    &  
    \AxiomC{$\Gamma \vdash F$}
    \AxiomC{$\Gamma \vdash G$}
    \RightLabel{$\wedge R$}
    \BinaryInfC{$\Gamma \vdash F \wedge G$}
    \DisplayProof
    \\&&\\ 
    \AxiomC{$\Gamma \vdash F_i$}
    \RightLabel{$\vee R$}
    \UnaryInfC{$\Gamma \vdash F_1 \vee F_2$}
    \DisplayProof
    &
    \AxiomC{$\Gamma; F \vdash G$}
    \RightLabel{$\supset R$}
    \UnaryInfC{$\Gamma \vdash F \supset G$}
    \DisplayProof
    &
    \AxiomC{$\Gamma_1 \vdash F$}
    \AxiomC{$\Gamma_2 \vdash G$}
    \RightLabel{$* R$}
    \BinaryInfC{$\Gamma_1, \Gamma_2 \vdash F * G$}
    \DisplayProof
    \\&&\\
    \AxiomC{$\Gamma, F \vdash G$}
    \RightLabel{$\text{\wand} R$}
    \UnaryInfC{$\Gamma \vdash F \text{\wand} G$}
    \DisplayProof
    & 
    \AxiomC{$\Gamma(\Gamma_1) \vdash H$}
    \RightLabel{Wk L}
    \UnaryInfC{$\Gamma(\Gamma_1; \Gamma_2) \vdash H$}
    \DisplayProof
    &
    \AxiomC{$\Gamma(\Gamma_1; \Gamma_1) \vdash H$}
    \RightLabel{Ctr L}
    \UnaryInfC{$\Gamma(\Gamma_1) \vdash H$}
    \DisplayProof 
    \\&&\\
    \AxiomC{$\Gamma(\Gamma_1; \top) \vdash H$} 
    \RightLabel{$EqAnt_1$} 
    \doubleLine\dottedLine 
    \UnaryInfC{$\Gamma(\Gamma_1) \vdash H$} 
    \DisplayProof 
    &
    \AxiomC{$\Gamma(\Gamma_1, \mtop) \vdash H$} 
    \RightLabel{$EqAnt_2$}  
    \doubleLine
    \dottedLine
    \UnaryInfC{$\Gamma(\Gamma_1) \vdash H$} 
    \DisplayProof
  \end{tabular}
  }
  \caption{\LBI: a {\BI } sequent calculus. 
  Inference rules with a double-dotted line are bidirectional. 
  {\small $i \in \{1,2\}$}. 
  Structural connectives are fully associative and commutative.}
  \label{LBI_calculus}
\end{figure*}
 
A {\BI } formula {\small $F(, G, H)$} with or without 
a sub-script is constructed from 
the following grammar:
 {\small $F := p \ | \ \top \ | \ \abot \ | \ \mtop \ | \ 
F \wedge F \ | \ F \vee F \ | \ F {\supset} F \ |$
$F * F \ | \ 
F {\text{\wand}} F$}. 
The set of {\BI} formulas is denoted by {\small $\mathfrak{F}$}. 
\hide{A {\BI} formula {\small $\alpha$} is defined by:\\
\indent {\small $\alpha := F \ | \ \emptyset \ | \ \O$}\\
where {\small ``$\emptyset$''} (resp. {\small ``$\O$''}) denotes a unary additive 
(resp. multiplicative) structural 
connective which acts as a proxy for {\small ``$\top$''} (resp. 
for {\small ``$\top^*$''}).} 
\begin{definition}[{\BI} structures]
{\BI } structure {\small $\Gamma (,  Re)$} with 
or without a sub-/super-script, 
commonly 
referred to as a bunch \cite{OHearnP99}, is 
defined by:
 {\small $\Gamma :=  F \ | \ \Gamma; \Gamma \ | \ 
     \Gamma, \Gamma$}.  We denote 
by {\small $\mathfrak{S}$} the set of 
{\BI} structures. 
\end{definition} 
For binding order,  
{\small [$
\wedge, \vee, *] \gg [\supset, \text{\wand}] \gg
[; \: ,] \gg [\forall\quad  \exists] \gg [\neg^{\dagger}] \gg [\wedge^{\dagger}, \vee^{\dagger}] \gg [ 
\rightarrow^{\dagger}, \leftrightarrow^{\dagger}]$} in a decreasing precedence. Connectives in the same group have the same 
precedence. \\
\indent Both of the structural connectives ``;'' and ``,'' are defined to be fully associative and 
commutative, and we assume as such everywhere 
we talk about {\BI} structures. On the other hand, 
we do not assume distributivity of ``;'' over `,' or 
vice versa. 
A context ``{\small $\Gamma(-)$}'' (with a hole ``{\small $-$}") takes the form of 
a tree because of the nesting of
additive/multiplicative structures.  
\begin{definition}[Context]
  A context {\small $\Gamma(-)$} is finitely constructed 
  from the following grammar: \\
    \indent {\small $\Gamma(-) := - \ | \ -; \Gamma \ | \ \Gamma; - \ | \
    -, \Gamma \ | \ \Gamma, - \ | \ \Gamma(-); \Gamma \ | \  
    \Gamma; \Gamma(-) \ | \ 
      \Gamma(-), \Gamma \ | \ \Gamma, \Gamma(-)$}. \\  
 Given any context {\small $\Gamma_1(-)$} and 
 any {\small $\Gamma_2 \in \mathfrak{S}$}, we assume that {\small $\Gamma_1(\Gamma_2)$} 
 is some {\BI} structure {\small $\Gamma_3$} 
 such that {\small $\Gamma_3 = \Gamma_1(\Gamma_2)$}. 
\end{definition}    
\begin{definition}[Sequents]   
  The set of {\BI} sequents {\small $\mathfrak{D}$} 
  is defined by: \\
  \indent {\small $\mathfrak{D} := 
  \{\Gamma \vdash F \ | \ \Gamma \in \mathfrak{S} \wedge^{\dagger} 
  F \in \mathfrak{F}\}$}.\\ 
  The left hand side of {\small $\vdash$} is termed 
  antecedent, and the right hand side of {\small $\vdash$} 
  consequent. 
\end{definition}
A variant of the first {\BI} sequent calculus {\LBI} is found in 
Figure \ref{LBI_calculus}. Notice how already 
we do not consier the nullary structural connectives. 
All the additive inference rules 
share contexts, \emph{e.g.} in {\small $\vee L$} the same 
context in the conclusion propagates onto both premises. 
Multiplicative inference rules are 
context-free \cite{351148} or resource sensitive. A good example to 
illustrate this is {\small $* R$}: both 
{\small $\Gamma_1$} and {\small $\Gamma_2$} in  
the conclusion sequent are viewed 
as resources for the inference rule, and are split into the 
premises of the rule. Note again our assumption of the full commutativity 
of ``,'' here. 
{\Cut} is admissible in \LBI.
\begin{lemma}[Cut admissibility in LBI]
  There is a direct cut elimination procedure which 
  proves admissibility of {\Cut} in {\LBI} (sketched in \cite{Pym02}; corrected in \cite{Arisaka2012-4}). 
\end{lemma} 
\hide{
\subsection{On representation of {\LBI} structures} 
A fine distinction between 
additive/multiplicative structures is
often encapsulated as a detail in coherent equivalence \cite{Pym02}. 
However, an arbitrary choice of a representation of {\BI} structures  
has a considerable downside of masking the semantically natural 
viewpoint about them, which is to view structures
as nestings of not individual formulas 
but of additive and multiplicative 
{\it structural layers}. \hide{care only 
about the focused principal formulas (\emph{Cf.} 
Figure \ref{LBI_calculus}). 
The principal formula of an inference 
rule is a formula in 
the conclusion sequent which is active for the rule. For example, in 
{\small $* L$}: 
{\small 
\begin{center}
   \AxiomC{$\Gamma(F, G) \vdash H$} 
   \RightLabel{$* L$} 
   \UnaryInfC{$\Gamma(F * G) \vdash H$} 
   \DisplayProof
\end{center}  
}
\noindent , it is the {\small ``$F * G$''} focused.  
From the mere aspect of formulation of {\LBI} inference 
rules, it indeed matters not,  in most of the rules, 
how a {\LBI} structure is represented.
Such a statement as  
``one particular representaion of structures is just as 
significant as all the others'' may appear to be justifiable.
One exception, however, 
is {\small $\rightarrow L$}. Suppose a simple {\LBI} sequent  
{\small $(\Gamma_1; \Gamma_2; F {\rightarrow} G), \Gamma_3 \vdash H$}. 
Which of the following should be ``correct'' when, 
in a bakckward proof search, $\rightarrow L$ applies?  
{\small 
\begin{center}
   \AxiomC{$\Gamma_1; \Gamma_2 \vdash F$} 
   \AxiomC{$(\Gamma_1; \Gamma_2; G), \Gamma_3 \vdash H$} 
   \RightLabel{$\rightarrow L_1$} 
   \BinaryInfC{$(\Gamma_1; \Gamma_2; F {\rightarrow} G), 
   \Gamma_3 \vdash H$} 
   \DisplayProof 
   {\ }\\ {\ }\\{\ }\\
   \AxiomC{$\Gamma_1 \vdash F$} 
   \AxiomC{$(\Gamma_1; \Gamma_2; G), \Gamma_3 \vdash H$} 
   \RightLabel{$\rightarrow L_2$} 
   \BinaryInfC{$(\Gamma_1; \Gamma_2; F {\rightarrow} G), \Gamma_3 
   \vdash H$} 
   \DisplayProof 
   {\ }\\{\ }\\{\ }\\ 
   \AxiomC{$\O_a \vdash F$} 
   \AxiomC{$(\Gamma_1; \Gamma_2; G), \Gamma_3 \vdash H$} 
   \RightLabel{$\rightarrow L_3$} 
   \BinaryInfC{$(\Gamma_1; \Gamma_2; F {\rightarrow} G), \Gamma_3 \vdash
   H$} 
   \DisplayProof
\end{center}
} 
\noindent Though the third is in general unsuitable, 
earlier studies do not necessarily  agree on either 
of the first two. Since both contraction and weakening in {\LBI} are 
an explicit structural rule, if {\small $\rightarrow L_1$} 
is derivable, then so is {\small $\rightarrow L_2$} and vice versa. 
However, it is not as much on derivability of one choice from another
that we intended in the posed question as on 
the degree of conformity of syntax to semantics.\\
\indent 
}
\begin{definition}[A {\BI} structure in
  nested structural layers]
  An antecedent structure {\small $\Gamma$} in 
  nested structural layers is defined as follows:     
  {\small  
  \begin{eqnarray}\nonumber
\Gamma &:=& F \ | \ \mathcal{M} \ | \ \mathcal{A}\\\nonumber
    \mathcal{M} &:=& F, \mathcal{M}' \ | \ 
    \mathcal{A}, \mathcal{M}'\\\nonumber
    \mathcal{M}' &:=& F \ | \ \mathcal{A} \ | 
    \ F, \mathcal{M}' \ | \ \mathcal{A}, \mathcal{M}'\\\nonumber
    \mathcal{A} &:=& F; \mathcal{A}' \ | \ 
    \mathcal{M}; \mathcal{A}'\\\nonumber
    \mathcal{A}' &:=& F \ | \ \mathcal{M} \ | 
    \ F; \mathcal{A}' \ | \ \mathcal{M}; \mathcal{A}'
  \end{eqnarray}  
  }
  Each of the {\small $\mathcal{A}$} (resp. 
  {\small $\mathcal{M}$}) substructures of {\small $\Gamma$}
  is termed an additive (resp. multiplicative) structural 
  layer.
\end{definition}    
By this definition we can recognise 
the boundary 
between additives and multiplicatives (in which incidentally lies
the distinct logical characteristics of \BI).  
An example of a {\BI} structure in nested structural layers
is found in Figure \ref{figure2}.  
\begin{figure}[!h]
  \centering
  \scalebox{0.85}{
\begin{tikzpicture}
  \node (root) at (4,0) {$\mathcal{M}_0$};
  \node (1l1) at (2,-1) {$F_1$};
  \node (1l2) at (4,-1) {$F_6$};
  \node (1l3) at (6,-1) {$\mathcal{A}$};
  \node (2l1) at (4, -2) {$F_2$};
  \node (2l2) at (6, -2) {$F_5$};
  \node (2l3) at (8, -2) {$\mathcal{M}_1$};
  \node (3l1) at (6, -3) {$F_3$};
  \node (3l2) at (10, -3) {$F_4$};
  \draw (root) -- (1l1);
  \draw (root) -- (1l2);
  \draw (root) -- (1l3);
  \draw (1l3) -- (2l1);
  \draw (1l3) -- (2l2);
  \draw (1l3) -- (2l3);
  \draw (2l3) -- (3l1);
  \draw (2l3) -- (3l2);
\end{tikzpicture}
}
\caption{{\small $F_1, ((F_3, F_4); F_2; F_5), F_6$} as represented 
in nested structural layers.}
\label{figure2}
\end{figure}

There are two multiplicative structural layers: {\small ``$F_1, F_6, \mathcal{A}$''} and {\small ``$F_3, F_4$''}; and one additive structural layer
{\small ``$F_2; F_5; \mathcal{M}_1$''}, with 
{\small $\mathcal{A}$} denoting {\small ``$F_2; F_5; \mathcal{M}_1$''} and 
{\small $\mathcal{M}_1$} denoting {\small ``$F_3, F_4$''}. 
For any structure in which 
two structural layers nest, the structural layer holding the other 
structural layer within is described as the outer structural layer 
of the two, while that enclosed in the other is described as the inner 
structural layer. 
}
\section{\LBIN: A Structural-Rule-Free BI Sequent Calculus}     
In this section we present a new {\BI} sequent calculus  
{\LBIN} (Figure \ref{LBIN_calculus}) 
in which 
no structural rules appear. 
We first introduce 
notations that are necessary to read inference rules 
in the calculus. First, from now on, whenever we write 
{\small $\widetilde{\Gamma}$} for any {\BI} structure, 
we indicate that it may be empty. The emptiness is in the following sense: 
{\small $\widetilde{\Gamma_1}; \Gamma_2 = \Gamma_2$} if 
{\small $\Gamma_1$} is empty; 
and {\small $\widetilde{\Gamma_1}, \Gamma_2 = \Gamma_2$}
if {\small $\Gamma_1$} is empty. Apart from this, 
we use two other notations.    
\subsection{Essence of antecedent structures} 
Co-existence of {\IL} and {\MILL} in {\BI} calls for new 
contraction-absorption techniques. 
Possible interferences to one structural rule from 
the others need considered. 
To illustrate the technical difficulty, 
{\small $EqAnt_{2\: \LBI}$} for instance interacts directly with 
{\small $Wk L_{\LBI}$}. When {\small $Wk L_{\LBI}$} is absorbed 
into the rest, the effect propagates to one direction of \linebreak
{\small $EqAnt_{2\:\LBI}$}, resulting in; 
{\small 
\begin{center}
  \AxiomC{$\Gamma(\Gamma_1) \vdash H$} 
  \RightLabel{$EA_2$} 
  \UnaryInfC{$\Gamma(\Gamma_1, ({\mtop}; \widetilde{\Gamma_2})) \vdash H$}  
  \DisplayProof
\end{center}
} 
\noindent Hence absorption of {\small $Wk L_{\LBI}$} must 
involve analysis of {\small $EqAnt_{2\:\LBI}$} as well. 
To solve this particular problem we define 
a new notation of `essence' of {\BI} structures. 
\begin{definition}[Essence of {\BI} structures]   
    Let {\small $\Gamma_1$} be a {\BI} structure.  
    Then we have a set of its essences as defined 
    in the following 
    inductive rules. 
    \begin{multicols}{2} 
    \begin{itemize} 
        \item {\small $\Gamma_2$} is an essence of 
            {\small $\Gamma_1$} if {\small $\Gamma_1 = \Gamma_2$}.\footnote{For some {\small $\Gamma_2$}. The equality is 
                of course up to associativity and commutativity.}  
        \item {\small $\Gamma(\Gamma', (\mtop; \widetilde{\Gamma_2}))$}\footnote{For some {\small $\widetilde{\Gamma_2}$}; 
                similarly in the rest.} 
            is an essence of {\small $\Gamma_1$} if 
            {\small $\Gamma(\Gamma')$} is 
            an essence of {\small $\Gamma_1$}. 
            \item {\small $\Gamma((\Gamma', ({\mtop}; \widetilde{\Gamma_2})); \Gamma'')$} is 
                an essence of {\small $\Gamma_1$} 
                if {\small $\Gamma(\Gamma'; \Gamma'')$} 
                is an essence of {\small $\Gamma_1$}.  
            \end{itemize}
\end{multicols} 
By {\small $\mathbb{E}(\Gamma_1)$} we denote 
an essence of {\small $\Gamma_1$}. 
\end{definition}  
\begin{figure*}[!t]
  \begin{center} 
    \scalebox{1}{ 
    \begin{tabular}{cccc} 
      \IDTH & \Lbot & \hspace{-1cm} \Rtop &
      \AxiomC{}
      \RightLabel{$\mtop R$} 
      \UnaryInfC{$\mathbb{E}(\widetilde{\Gamma}; \mtop) \vdash {\mtop}$} 
      \DisplayProof
      \\&&&\\
     \multicolumn{2}{c}{\Lwedge} & \multicolumn{2}{c}{\Rwedge}\\&&&\\
      \multicolumn{2}{c}{\Lvee} &
      \multicolumn{2}{c}{\Rvee}
      \\&&&\\ \multicolumn{2}{c}{\LrightarrowTH}&
      \multicolumn{2}{c}{\RrightarrowTH}\\&&&\\ 
      \multicolumn{2}{c}{\Lstar} & \multicolumn{2}{c}{
	   \AxiomC{$Re_i \vdash F_1$} 
	   \AxiomC{$Re_j \vdash F_2$} 
	   \RightLabel{$* R$} 
	   \BinaryInfC{$\Gamma' \vdash F_1 * F_2$} 
	   \DisplayProof}\\&&&\\
\multicolumn{3}{c}{
\AxiomC{$Re_i \vdash F$} 
\AxiomC{$\Gamma((\widetilde{Re_j}, G); (\widetilde{\Gamma'}, 
    \mathbb{E}(\widetilde{\Gamma_1}; F {\text{\wand}} G))) \vdash H$}
\RightLabel{$\text{\wand} L$} 
\BinaryInfC{$\Gamma(\widetilde{\Gamma'}, 
    \mathbb{E}(\widetilde{\Gamma_1}; F {\text{\wand}} G)) \vdash H$} 
\DisplayProof 
}
& \Rstararrow
\hide{
      \multicolumn{3}{c}{\LstararrowTHO}\\&&\\
      \multicolumn{3}{c}{\LstararrowTHT}\\&&\\
     \multicolumn{3}{c}{\LstararrowTHTH}\\&&\\
     \multicolumn{3}{c}{\LstararrowTHF} 
     }
     \hide{
     \multicolumn{3}{c}{ 
\AxiomC{$\Gamma(\Gamma_1) \vdash H$} 
      \RightLabel{$EA_2$} 
      \UnaryInfC{$\Gamma(\Gamma_1, (\O_m; \Gamma_2)) \vdash H$} 
      \DisplayProof } 
      } 
    \end{tabular}
     }
  \end{center}
  \caption{\LBIN: a {\BI } sequent calculus with zero occurrence of
      explicit structural rules. {\small $i, j\in \{1,2\}$}. 
      {\small $i \not=j$}. Structural connectives are 
      fully associative and commutative. In 
      {\small $* R$} and {\small ${\text{\wand}} L$}, 
      if {\small $\Gamma'$} is not empty, 
      {\small $(Re_1, Re_2) \in \code{Candidate}(\Gamma')$}; 
      otherwise, {\small $Re_i = {\mtop}$} and 
      {\small $Re_j$} is empty. Both 
      {\small $\mathbb{E}$} and \code{Candidate} are as defined 
      in the main text.}
  \label{LBIN_calculus}
\end{figure*}

The essence takes care of an arbitrary number of {\small $EA_2$} applications, while nicely retaining  
a compact representation of a sequent (see the calculus).  
In 
each of {\small $\supset L$} and {\small ${\text{\wand}} 
L$}, the essence in the premise(s) 
and that in the conclusion are the same and identical {\BI} structure.  
Specifically, in a derivation tree, the use of 
{\small $\mathbb{E}(\Gamma)$} in multiple 
sequents in the derivation tree signifies the same {\BI} structure. 
\begin{example}
   Given a \LBIN-derivation:  
   {\small 
   \begin{center}    
     \AxiomC{}
     \RightLabel{$id$}
     \UnaryInfC{$F_1; ((\mtop; \Gamma_1), F_1 {\supset} F_2) \vdash F_1$}    
     \AxiomC{}
     \RightLabel{$id$}
     \UnaryInfC{$F_2; F_1; ((\mtop; \Gamma_1), F_1 {\supset} F_2)
     \vdash F_2$}
     \RightLabel{$\supset L$}
      \BinaryInfC{$F_1; ((\mtop; \Gamma_1), 
      F_1 {\supset} F_2) \vdash F_2$}
      \DisplayProof
   \end{center}
   }  
   it can be alternatively written down by; 
   \begin{center}
     \scalebox{0.9}{  
     \AxiomC{}
     \RightLabel{$id$} 
     \UnaryInfC{$\mathbb{E}(F_1; F_1 \supset F_2) \vdash F_1$} 
     \AxiomC{}
     \RightLabel{$id$} 
     \UnaryInfC{$F_2; \mathbb{E}(F_1; F_1 \supset F_2) \vdash 
     F_2$} 
     \RightLabel{$\supset L$}  
     \BinaryInfC{$\mathbb{E}(F_1; F_1 \supset F_2) \vdash F_2$} 
     \DisplayProof
     }
   \end{center}  
   \indent\indent where {\small $\mathbb{E}(F_1; F_1 \supset F_2) = 
   F_1; ((\mtop; \Gamma_1), F_1 \supset F_2)$}. 
\end{example}  
{\small $\mathbb{E}'(\Gamma)$} (or 
{\small $\mathbb{E}_1(\Gamma)$} or any essence 
that differs from {\small $\mathbb{E}$} by the presence 
of a sub-script, a super-script or both)
in the same 
derivation tree does not have to be coincident 
with the {\BI} structure that the {\small $\mathbb{E}(\Gamma)$} 
denotes. However, we do - for prevention of inundation of many 
super-scripts and sub-scripts - make an exception. In 
the cases where no ambiguity is likely to arise such as 
in the following; 
\begin{center}
  \scalebox{0.9}{ 
  \AxiomC{$\Gamma(\mathbb{E}(\Gamma_1; F; G)) \vdash H$} 
  \RightLabel{$\wedge L$} 
  \UnaryInfC{$\Gamma(\mathbb{E}(\Gamma_1; F \wedge G)) \vdash 
  H$} 
  \DisplayProof
  } 
\end{center}
we assume that the essence in the conclusion is 
the same antecedent structure as the essence in 
the premise(s) except what the inference rule 
modifies.

\subsection{Correspondence between {\small $Re_i$/$Re_j$} and {\small
$\Gamma'$}}  
\begin{definition}[Relation {\small $\preceq$}]
  We define a reflexive and transitive binary 
  relation {\small $\preceq: 
      \mathfrak{S} \times \mathfrak{S}$} as follows.  
  \begin{multicols}{2}
  \begin{itemize} 
      \item {\small $\Gamma_1 \preceq \Gamma_2$} 
          if {\small $\Gamma_1 = \Gamma_2$}. 
      \item {\small $\Gamma(\Gamma_1) \preceq \Gamma(\Gamma_1; \Gamma')$}. 
      \item {\small $[\Gamma_1 \preceq \Gamma_2] \wedge^{\dagger} 
              [\Gamma_2 \preceq \Gamma_3] 
              \rightarrow^{\dagger} [\Gamma_1 \preceq \Gamma_3]$}.  
  \end{itemize} 
  \end{multicols}   
\end{definition}  
Intuitively if {\small $\Gamma_1 \preceq \Gamma_2$}, then there 
exists a \LBI-derivation:  
\begin{center} 
    \scalebox{0.9}{ 
        \AxiomC{$\Gamma(\Gamma_1) \vdash H$} 
        \RightLabel{$Wk L$}  
        \doubleLine 
        \UnaryInfC{$\Gamma(\Gamma_2) \vdash H$} 
        \DisplayProof
    }
\end{center}
for any {\small $\Gamma(-)$} and any {\small $H$}. 
Here and elsewhere a double line indicates zero or more derivation steps. 
\begin{definition}[Candidates] 
    Let {\small $\Gamma$}  be a {\BI} structure, 
    then any of the following pairs is 
    a candidate of {\small $\Gamma$}. 
    \begin{multicols}{2} 
    \begin{itemize} 
        \item {\small $(\Gamma_x, {\mtop})$}
            if {\small $\Gamma_x \preceq \Gamma$}.  
        \item {\small $(\Gamma_x, \Gamma_y)$} 
            if {\small $\Gamma_x, \Gamma_y \preceq \Gamma$}. 
    \end{itemize}
\end{multicols}   
We denote the set of candidates of {\small $\Gamma$} 
by {\small $\code{Candidate}(\Gamma)$}. 
\end{definition}   
Now we see the connection between {\small $Re_i/Re_j$} and 
{\small $\Gamma'$} in the two rules {\small $* R/{\text{\wand}} L$}.
\begin{definition}[$Re_i/Re_j$ in $* R/{\text{\wand}} L$]   
    In {\small $* R$} and {\small ${\text{\wand}} L$}, 
    if {\small $\Gamma'$} is empty,\footnote{This case 
        applies to {\small ${\text{\wand}} L$} only.} 
    {\small $Re_i = {\mtop}$} and 
    {\small $Re_j$} is empty. If it is not empty, 
    then 
    {\small $(Re_1, Re_2) \in \code{Candidate}(\Gamma')$}. 
\end{definition}    
Let us reflect on the purposes of the two notations that 
we have introduced. An essence absorbs a finite number of 
EA$_2$ derivation steps. \code{Candidate} absorbs 
a finite number of $Wk$ derivation steps. Then 
what the inference rules in {\LBIN} are doing should be 
clear. There are no structural rules. Implicit contraction 
occurs only in {\small $\supset L$} and 
{\small ${\text{\wand}} L$}.\footnote{Implicit weakening 
and others occur also in other inference rules; but they 
are not very relevant in backward theorem proving.}
In both of the inference rules, 
a structure than a formula duplicates upwards. This is 
necessary, for we have the following observation. 
\begin{observation}[Structural contractions are not 
    admissible]{\ } \\ 
  There exist sequents 
  {\small $\Gamma \vdash F$} which are derivable 
  in {\LBI} - {\Cut} 
  but not derivable in {\LBI} - {\Cut} without structural contraction. 
  \label{underivable_without_StrCtr}
\end{observation} 
\begin{proof} 
  For {\small $\text{\wand} L$} use a sequent 
  {\small $\top {\text{\wand}}
  p_1, \top {\text{\wand}} (p_1 {\supset} p_2) 
  \vdash p_2$} and assume that 
  every propositional variable is distinct. Then without 
  contraction, there are several derivations.  
  Two sensible ones are shown below (the rest similar). Here 
  and elsewhere we may label a sequent by {\small $D$} 
  with or without a sub-/super-script just so that we 
  may refer to it by the name.\\ 
  \begin{enumerate}
    \item  
      \scalebox{0.9}{
	\AxiomC{}
	\RightLabel{$\top R$}
	\UnaryInfC{$\top {\text{\wand}} (p_1 {\supset} p_2) 
	\vdash \top$}
	\AxiomC{$p_1 \vdash p_2$}
	\RightLabel{$\text{\wand} L$}
	\BinaryInfC{$D: \top {\text{\wand}} p_1, 
	\top \text{\wand} (p_1 {\supset} p_2) \vdash p_2$}
	\DisplayProof}{\ }\\\\
    \item 
      \scalebox{0.9}{
	\AxiomC{}
	\RightLabel{$\top R$}
	\UnaryInfC{$\top {\text{\wand}} p_1 \vdash \top$} 
	\AxiomC{$\top \vdash p_1$}   
	\AxiomC{}      
	\RightLabel{$id$}   
	\UnaryInfC{$p_2 \vdash p_2$}
	\RightLabel{$\supset L$}    
	\BinaryInfC{$\top; p_1 {\supset} p_2 \vdash p_2$}
	\RightLabel{$EqAnt_1 L$}
	\UnaryInfC{$p_1 {\supset} p_2 \vdash p_2$} 
	\RightLabel{$\text{\wand} L$}
	\BinaryInfC{$D: \top {\text{\wand}} p_1, 
	\top \text{\wand} (p_1 {\supset} p_2)
	\vdash p_2$}
	\DisplayProof}{\ }\\
  \end{enumerate}  
  In both of the derivation trees above, one branch is open. Moreover,
  such holds true when only formula-level contraction is permitted in
  \LBI. The sequent 
  {\small $D$} cannot be derived under the given restriction. 
  In the presence of structural contraction, however, 
  another construction is possible:
  \begin{center}        
  \scalebox{0.9}{  
  \AxiomC{$\Pi(D_1)$}
   \AxiomC{$\Pi(D_2)$}
\RightLabel{$\text{\wand} L$}
    \BinaryInfC{$(\top \text{\wand} p_1,  
    \top \text{\wand} (p_1 \supset p_2)); 
    (\top \text{\wand} p_1, 
    \top \text{\wand} (p_1 \supset p_2)) \vdash p_2$} %
    \RightLabel{$Ctr L$}
    \UnaryInfC{$D: \top \text{\wand} p_1, 
    \top \text{\wand} (p_1 \supset p_2) 
    \vdash p_2$}
    \DisplayProof 
    }  
   \end{center}
   where {\small $\Pi(D_1)$} and {\small $\Pi(D_2)$} are:  
    \begin{center}
    \begin{description}
      \item[{\small $\Pi(D_1)$}: ]{\ }\\
	\begin{center}  
	  \scalebox{0.9}{
	  \AxiomC{}
	    \RightLabel{$\top R$}
	    \UnaryInfC{$\top \text{\wand} 
	    (p_1 \supset p_2) \vdash \top$} 
	    \DisplayProof 
	    }
	\end{center} 
      \item[{\small $\Pi(D_2)$}: ]{\ }\\ 
	\begin{center} 
	  \scalebox{0.9}{\hspace{-1.7cm}
	  \AxiomC{}  
	  \RightLabel{$\top R$}
	  \UnaryInfC{$\top \text{\wand} p_1 \vdash \top$}
	  \AxiomC{}
	  \RightLabel{$id$}   
	  \UnaryInfC{$p_1 \vdash p_1$} 
	  \AxiomC{} 
	  \RightLabel{$id$}   
	  \UnaryInfC{$p_2 \vdash p_2$}
	  \RightLabel{$Wk L$}
	  \UnaryInfC{$p_1; p_2 \vdash p_2$} 
	  \RightLabel{$\supset L$} 
	  \BinaryInfC{$p_1; p_1 \supset p_2 
	  \vdash p_2$}
	  \RightLabel{$\text{\wand} L$}
	  \BinaryInfC{$p_1; (\top \text{\wand} 
	  (p_1 \supset p_2)) \vdash p_2$}
	  \DisplayProof 
	  }
	\end{center}
    \end{description}
  \end{center}
  where all the derivation tree branches are closed upward. \\
  \indent For {\small $\supset L$}, use {\small 
  $(\mtop;p_1),
  (\mtop; p_1 {\supset} p_2)
  \vdash p_2$}. Without structural contraction we have (only 
  two sensible ones are shown; the rest similar):
  \begin{enumerate}
    \item{\ }  
      {\small 
      \begin{center}        
	\AxiomC{$\mtop \vdash p_1$}    
	\AxiomC{}
	\RightLabel{$id$}
	\UnaryInfC{$p_2 \vdash p_2$}
	\RightLabel{$Wk L$}
	\UnaryInfC{$\mtop; p_2 \vdash p_2$}
        \RightLabel{$EA_2$}
	\UnaryInfC{$(\mtop; p_1), (\mtop; p_2) \vdash p_2$}
	\RightLabel{$\supset L$}
	\BinaryInfC{$D: (\mtop; p_1 ),
	(\mtop; p_1 {\supset} p_2)
	 \vdash p_2$} 
	 \DisplayProof
      \end{center}
      } 
    \item{\ }  
      {\small 
      \begin{center}        
	\AxiomC{$p_1 \vdash p_2$}
	\RightLabel{$Wk L$}
	\UnaryInfC{$\mtop; p_1 \vdash p_2$}
	\RightLabel{$EA_2$}
	\UnaryInfC{$D: (\mtop; p_1), (\mtop; p_1 {\supset} p_2) 
	\vdash p_2$}
	\DisplayProof
      \end{center}
      }
  \end{enumerate}
  In the presence of structural contraction, there is a closed derivation. 
  \begin{center}       
    \scalebox{0.9}{   
    \AxiomC{}
    \RightLabel{$id$}
    \UnaryInfC{$p_1 \vdash p_1$}
    \RightLabel{$Wk L$}
    \UnaryInfC{$\mtop; p_1; \mtop \vdash p_1$}
    \AxiomC{}
    \RightLabel{$id$}
    \UnaryInfC{$p_2 \vdash p_2$}
    \RightLabel{$Wk L$}
    \UnaryInfC{$\mtop; p_1; \mtop; p_2 \vdash p_2$}
    \RightLabel{$\supset L$}  
    \BinaryInfC{$
    \mtop; p_1; \mtop; p_1 {\supset} p_2 \vdash p_2$}
    \doubleLine 
    \RightLabel{$EA_2$}
    \UnaryInfC{$((\mtop; p_1), (\mtop; p_1 {\supset} p_2));
    ((\mtop; p_1), (\mtop; p_1 {\supset} p_2))
     \vdash p_2$}
    \RightLabel{$Ctr L$}
    \UnaryInfC{$D: (\mtop; p_1),
    (\mtop; p_1 {\supset} p_2) \vdash p_2$}
    \DisplayProof 
    }
  \end{center}  
  \hide{and assume that 
  every propositional variable is distinct. Then without 
  contraction, there are several derivations of which 
  two sensible ones are shown below.\\
  \begin{enumerate}
    \item  
      \scalebox{0.88}{
	\AxiomC{}
	\RightLabel{$\top R$}
	\UnaryInfC{$p_3, \top {\text{\wand}} (p_1 {\rightarrow} p_2) 
	\vdash \top$}
	\AxiomC{$p_5; p_1 \vdash p_2$}
	\RightLabel{$\text{\wand} L$}
	\BinaryInfC{$D: p_5; (\top {\text{\wand}} p_1, 
	\top \text{\wand} (p_1 {\rightarrow} p_2), p_3) \vdash p_2$}
	\DisplayProof}{\ }\\\\
    \item 
      \scalebox{0.88}{
	\AxiomC{}
	\RightLabel{$\top R$}
	\UnaryInfC{$p_3, \top {\text{\wand}} p_1 \vdash \top$} 
	\AxiomC{$p_5 \vdash p_1$}   
	\AxiomC{}      
	\RightLabel{$id$}   
	\UnaryInfC{$p_2 \vdash p_2$}
	\RightLabel{$Wk L$}
	\UnaryInfC{$p_5; p_2 \vdash p_2$}
	\RightLabel{$\rightarrow L$} 
	\BinaryInfC{$p_5; p_1 {\rightarrow} p_2 \vdash p_2$} 
	\RightLabel{$\text{\wand} L$}
	\BinaryInfC{$D: p_5; (\top {\text{\wand}} p_1, 
	\top \text{\wand} (p_1 {\rightarrow} p_2), p_3)
	\vdash p_2$}
	\DisplayProof}{\ }\\
  \end{enumerate}  
  In both of the derivation trees above, one branch is open. Moreover,
  such holds true when only formula-level contraction is permitted in
  \LBI. The sequent 
  {\small $D$} cannot be derived under the given restriction. 
  In the presence of structural contraction, however, 
  another construction is possible:
  \begin{center}        
  \scalebox{0.76}{  
  \AxiomC{$\Pi(D_1)$}
   \AxiomC{$\Pi(D_2)$}
\RightLabel{$\text{\wand} L$}
    \BinaryInfC{$p_5; (\top \text{\wand} p_1,  
    \top \text{\wand} (p_1 \rightarrow p_2), p_3, p_4); 
    (\top \text{\wand} p_1, 
    \top \text{\wand} (p_1 \rightarrow p_2), p_3, p_4) \vdash p_2$} %
    \RightLabel{$Ctr L$}
    \UnaryInfC{$D: p_5; (\top \text{\wand} p_1, 
    \top \text{\wand} (p_1 \rightarrow p_2), p_3, p_4) 
    \vdash p_2$}
    \DisplayProof 
    }  
   \end{center}
   where {\small $\Pi(D_1)$} and {\small $\Pi(D_2)$} are:  
    \begin{center}
    \begin{description}
      \item[{\small $\Pi(D_1)$}: ]{\ }\\
	\begin{center}  
	  \scalebox{0.88}{
	  \AxiomC{}
	    \RightLabel{$\top R$}
	    \UnaryInfC{$\top \text{\wand} 
	    (p_1 \rightarrow p_2), p_3 \vdash \top$} 
	    \DisplayProof 
	    }
	\end{center} 
      \item[{\small $\Pi(D_2)$}: ]{\ }\\ 
	\begin{center} 
	  \scalebox{0.88}{\hspace{-1.7cm}
	  \AxiomC{}  
	  \RightLabel{$\top R$}
	  \UnaryInfC{$\top \text{\wand} p_1, p_3 \vdash \top$}
	  \AxiomC{}
	  \RightLabel{$id$}   
	  \UnaryInfC{$p_1 \vdash p_1$}
	  \RightLabel{$Wk L$}
	  \UnaryInfC{$p_5; p_1 \vdash p_1$} 
	  \AxiomC{} 
	  \RightLabel{$id$}   
	  \UnaryInfC{$p_2 \vdash p_2$}
	  \RightLabel{$Wk L$}
	  \UnaryInfC{$p_5; p_1; p_2 \vdash p_2$} 
	  \RightLabel{$\rightarrow L$} 
	  \BinaryInfC{$p_5; p_1; p_1 \rightarrow p_2 
	  \vdash p_2$}
	  \RightLabel{$\text{\wand} L$}
	  \BinaryInfC{$p_5; p_1; (\top \text{\wand} 
	  (p_1 \rightarrow p_2), p_3) \vdash p_2$}
	  \DisplayProof 
	  }
	\end{center}
    \end{description}
  \end{center}
  where all the derivation tree branches are closed upward.  
  }
\end{proof}

\section{Main Properties of \LBIN}
In this section we show 
the main properties of \LBIN, \emph{i.e.} 
admissibility of weakening, that of {\small $EA_2$}, that of 
both {\small $EqAnt_{1\:\LBI}$} and 
 {\small $EqAnt_{2\:\LBI}$},
that of contraction, and its equivalence to 
\LBI. Cut is also admissible. We will refer to derivation 
depth very often. 
\begin{definition}[Derivation depth] 
    By {\small $\Pi(D)$} we denote a derivation tree 
    of a sequent {\small $D$}. We assume that 
    {\small $\Pi(D)$} is always closed: every derivation 
    branch of the tree has an empty sequent as the leaf node (the 
    premise of an axiom). For derivation depth, 
    let {\small $\Pi(D)$} be a derivation tree. Then 
    the derivation depth of {\small $D'$}, a node 
    in {\small $\Pi(D)$}, is:  
    \begin{itemize} 
        \item 1 if {\small $D'$} is the conclusion node of 
            of an axiom inference rule. 
        \item 1 + (derivation depth of {\small $D_1$}) 
            if {\small $\Pi(D')$} looks like: 
            \begin{center} 
                \scalebox{0.9}{ 
                    \AxiomC{$\Pi(D_1)$}  
                    \UnaryInfC{$D'$} 
                    \DisplayProof
                }
            \end{center} 
        \item 1 + (the larger of the derivation depths 
            of {\small $D_1$} and {\small $D_2$}) if 
            {\small $\Pi(D')$} looks like: 
            \begin{center}  
                \scalebox{0.9}{ 
                    \AxiomC{$\Pi(D_1)$} 
                    \AxiomC{$\Pi(D_2)$} 
                    \BinaryInfC{$D'$} 
                    \DisplayProof
                }
            \end{center}
    \end{itemize}
\end{definition} 
\hide{
To previous reviews: what is this essence or why 
does it have to be defined? The answer should 
become evident if one attempts to define 
{\small $id_{\LBITH}$} without any concept of 
the sort of the essence: the conclusion sequent 
would become syntactically inexpressible (infinitely 
long). This concept is in fact a nice one which reflects 
the semantics of the multiplicative unit. 
}
  \subsection{Weakening admissibility and $EA_2$ admissibility} 
  Admissibilities of both weakening and $EA_2$
  are proved depth-preserving. 
   This means in case of weakening that if a sequent {\small $\Gamma(\Gamma_1) \vdash H$} 
   is derivable with derivation depth of $k$, then 
   {\small $\Gamma(\Gamma_1; \Gamma_2) \vdash H$} 
   is derivable with derivation depth of $l$ such that $l \le k$.
   \begin{proposition}[{\LBIN} weakening admissibility]
  If a sequent {\small $D: \Gamma(\Gamma_1) \vdash F$} is \LBIN-derivable, 
   then so is 
   {\small $D': \Gamma(\Gamma_1; \Gamma_2) \vdash F$}, preserving 
   the derivation depth.
   \label{admissible_weakening_LBI3}
\end{proposition} 
\begin{proof}    
    By induction on derivation depth of {\small $D$}.
    Details are in Appendix A. \qed
\end{proof}      
\begin{proposition}[Admissibility of $EA_2$]
  If a sequent {\small $D: \Gamma(\Gamma_1) \vdash F$} 
  is \LBIN-derivable, then so is {\small $D': \Gamma(\mathbb{E}(\Gamma_1))
  \vdash F$}, preserving the derivation depth.   
  \label{admissible_EA2}
\end{proposition}
\begin{proof}  
    By induction on derivation depth of {\small $D$}. 
   If it is one, \emph{i.e.} {\small $D$} is the conclusion 
   sequent of an axiom, then so is {\small $D'$}.  
   Inductive cases are straightforward due to 
   a near identical proof approach to the weakening 
   admissibility proof (see Appendix A). \qed
\end{proof}
\subsection{Inversion lemma}
The inversion 
lemma below is important in simplification of the subsequent 
discussion. 
\begin{lemma}[Inversion lemma for \LBIN] {\ }   
    For the following sequent pairs, if the 
    sequent on the left is \LBIN-derivable at
    most with the derivation depth of $k$,
  then so is (are) the sequent(s) on the right.   
{\small
    \begin{eqnarray}\nonumber
      {\small \Gamma(F \wedge G) \vdash H}, &&
      {\small \Gamma(F; G) \vdash H}\\\nonumber
      {\small \Gamma(F_1 \vee F_2) \vdash H}, && 
      \text{both } {\small \Gamma(F_1) \vdash H} \text{ and }
      {\small \Gamma(F_2) \vdash H}\\\nonumber
      {\small \Gamma(F * G) \vdash H}, && 
      {\small \Gamma(F, G) \vdash H}\\\nonumber
      {\small \Gamma(\Gamma_1; \top) \vdash H}, && 
      {\small \Gamma(\Gamma_1) \vdash H}\\\nonumber
      {\small \Gamma(\Gamma_1, {\mtop}) \vdash H}, && 
      {\small \Gamma(\Gamma_1) \vdash H}\\\nonumber
      {\small \Gamma \vdash F \wedge G}, && 
      \text{both } {\small \Gamma \vdash F} \text{ and } 
      {\small \Gamma \vdash G}\\\nonumber
      {\small \Gamma \vdash F {\supset} G}, && 
      {\small \Gamma; F \vdash G}\\\nonumber
      {\small \Gamma \vdash {F \text{\wand} G}}, && 
      {\small \Gamma, F \vdash G} 
    \end{eqnarray}  
    }
        \hide{
    \begin{enumerate}
    \item If $\Gamma(F \wedge G) \vdash H$ is deducible within 
      a derivation depth $k$, then so is $\Gamma(F; G) \vdash H$. 
    \item If $\Gamma(F_1 \vee F_2) \vdash H$ is deducible 
      within a derivation depth $k$, then so are $\Gamma(F_1) \vdash H$ 
      and $\Gamma(F_2) \vdash H$.  
    \item If $\Gamma(\Gamma_1; F \rightarrow G) \vdash H$ is deducible 
      within a derivation depth $k$, then so is 
      $\Gamma(\Gamma_1; G) \vdash H$.  
    \item If $\Gamma(F * G) \vdash H$ is deducible within 
      a derivation depth $k$, then so is $\Gamma(F, G) \vdash H$.   
    \item If $\Gamma(\top) \vdash H$ is deducible within 
      a derivation depth $k$, then so is $\Gamma(\emptyset) \vdash H$.  
    \item If $\Gamma(\top^*) \vdash H$ is deducible within 
      a derivation depth $k$, then so is $\Gamma(\O) \vdash H$.  
    \item If $\Gamma \vdash F \wedge G$ is deducible within a 
      derivation depth $k$, then so are $\Gamma \vdash F$ and 
      $\Gamma \vdash G$. 
    \item If $\Gamma \vdash F \rightarrow G$ is deducible 
      within a derivation depth $k$, then so is $\Gamma; F \vdash G$. 
    \item If $\Gamma \vdash F \text{\wand} G$ is deducible 
      within a derivation depth $k$, then so is $\Gamma, F \vdash G$.  
  \end{enumerate} 
  }
  \label{LBITH_inversion_lemma} 
    \hide{ 
    DO THIS LATER.
    The following {\LBITH} inference rules are depth-preserving 
    invertible for all the premises: $\wedge L_{\LBITH} \quad 
    \vee L_{\LBITH} \quad $
    }
  \end{lemma}{\vspace{-0.5cm}
  \begin{proof} 
      By induction on derivation depth. Details are in Appendix B. 
  \end{proof}    
\hide{
    
    Following \ref{LBITH_inversion_lemma}, we 
  also
  develop the concept of normalisation via  
  {\small $EqAnt''_2$} and {\small $ \O L$}
  to be rid of unnecessary analysis
  overhead.
  \begin{definition}  Given a {\LBITH} sequent {\small $D_1: \Gamma \vdash F$}, its 
    {\small $\O L$}-normalised sequent {\small 
    ${\small \O L}\norm({\small D_1)}$}
    is defined such to satisfy:
  \begin{itemize}
    \item {\small ${\small D_1 \leadsto^*_{\O L} \O L} \norm({\small D_1})$}.
    \item there is no sequent {\small $D'_1$} such that  
      {\small ${\small \O L}\norm({\small D_1}) \small{\leadsto_{\O L} D'_1}$}.
  \end{itemize} 
  Its {\small $EqAnt''_2$}-normalised sequent 
  {\small ${\small EA_2}\norm({\small D_1})$} is defined such to satisfy:
  \begin{itemize}
    \item {\small ${\small D_1 \leadsto^*_{EqAnt''_2} EA_2}\norm({\small D_1})$}. 
    \item there is no sequent {\small $D'_1$} such that 
      {\small ${\small EA_2}\norm({\small D_1) \leadsto_{EqAnt''_2} D'_1}$}.\\
  \end{itemize} 
\end{definition}    
\begin{definition}[{\small $\O L$/$EqAnt''_2$}-normalised LBI3 sequents] 
  Given a {\LBITH} derivation of a sequent {\small $D$}, its 
  {\small $\O L$}-normalised and also {\small $EqAnt_2''$}-normalised 
  sequent is defined as {\small $\repeatnormalise(D)$} where 
  {\small $\repeatnormalise$} is a function defined in the following 
  algorithm, taking as its input 
  a sequent {\small $D_{in}$} and returning a sequent {\small $D_{out}$} as its output.
  \begin{algorithm} 
  \begin{algorithmic}[1]
    {\small \Function{\repeatnormalise}{$D_{in}$}\Comment{$D_{out}$} 
    \State $D_{out} \gets D_{in}$
    \State $D_{test} \gets D_{out}$\Comment{the below loop 
    terminates once no more changes take place to $D_{out}$}.
    \While{true}
       \State $D_{out} \gets  {\O} L\norm(D_{out})$ 
       \State $D_{out} \gets EA_2\norm(D_{out})$ 
       \If{$D_{out}$ is $D_{test}$} 
           \State \textbf{return} $D_{out}$ 
	   \Else 
	   \State $D_{test} \gets D_{out}$ 
\EndIf
    \EndWhile\label{euclidendwhile}
    \EndFunction
    }
  \end{algorithmic}
\end{algorithm}
\end{definition} 
\begin{definition}[Normalised LBI3 derivation] 
  Given a {\LBITH} derivation of a sequent {\small $D$}, 
  denoted as {\small $\Pi(D)$}, its normalised derivation 
  {\small $\norm(\Pi(D))$} is defined inductively as follows: 
  \begin{itemize}
    \item if the derivation depth of {\small $D$} is 1, then  {\small 
      $\norm(\Pi(D))$} 
      is {\small $\Pi(\repeatnormalise(D))$}.
    \item if {\small $D$} is the conclusion sequent of a single-premised
      inference rule \textbf{Inf} whose premise sequent is {\small $D_1$}, 
      then {\small $\norm(\Pi(D))$} is such that 
      \begin{center}
	{\small 
	\AxiomC{$\norm(\Pi(D_1))$} 
	\RightLabel{\textbf{Inf}} 
	\UnaryInfC{$\repeatnormalise(D)$} 
	\DisplayProof
	}
      \end{center} 
      \emph{i.e.} a {\LBITH} derivation whose conclusion is 
      {\small $\repeatnormalise(D)$} which has a sub-derivation {\small 
      $\norm(\Pi(D_1))$}
      whose conclusion sequent is the premise sequent of \textbf{Inf}.
    \item if {\small $D$} is the conclusion sequent of a two-premised
     inference rule \textbf{Inf} whose premise sequents 
     are {\small $D_1$} and {\small $D_2$}, then 
     {\small $\norm(\Pi(D))$} is such that 
     \begin{center} 
       {\small 
        \AxiomC{$\norm(\Pi(D_1))$} 
	\AxiomC{$\norm(\Pi(D_2))$} 
	\RightLabel{\textbf{Inf}} 
	\BinaryInfC{$\repeatnormalise(D)$} 
	\DisplayProof
	}
     \end{center} 
     \emph{i.e.} a {\LBITH} derivation whose conclusion is 
     {\small $\repeatnormalise(D)$} which has sub-derivations {\small 
     $\norm(\Pi(D_1))$}
     and {\small $\norm(\Pi(D_2))$} whose conclusion sequents 
     are the premise sequents of \textbf{Inf}. \\
  \end{itemize} 
\end{definition}  
\begin{proposition}  
  Given a sequent {\small $D: \Gamma \vdash F$},    
  a {\LBITH} derivation of {\small $D$}, denoted as {\small $\Pi(D)$}, is closed 
  if and only if its normalised {\LBITH} derivation {\small $\norm(\Pi(D))$} is.
  \label{equivalence_normalised_derivation}
\end{proposition} 
\begin{proof}
  Follows from Lemma \ref{LBITH_inversion_lemma}.
\end{proof} 
{\ }\\
Speaking in proximity, 
what Proposition \ref{equivalence_normalised_derivation} 
indicates is the admissibility of {\small $EqAnt''_{2}$} and 
also of {\small $\O L$} if {\LBITH} inference rules were 
appropriately modified to absorb these, which, in fact, 
is not so difficult to see. Though aware 
of this consequence, we shall keep  
both {\small $EqAnt''_2$} and {\small $\O L$}, to avoid an overly complex 
calculus with little gain. \underline{Hereafter, 
we consider only normalised derivations}.
}  
\hide{The last two imply that no arbitrary introduction 
of structural units upwards is needed in {\LBITH} backward derivations 
(\emph{c.f.} Proposition \ref{admissible_eqant_LBI3} below).} 
\subsection{Admissibility of {\small $EqAnt_{1,2}$}}  
\hide{
\begin{center} 
  {\small 
  \AxiomC{$\Gamma(\Gamma_1, \O_m) \vdash F$} 
  \RightLabel{$EqAnt'_2$} 
  \UnaryInfC{$\Gamma(\Gamma_1) \vdash F$} 
  \DisplayProof 
  }
\end{center} 
}
\begin{proposition}[Admissibility of {\small $EqAnt_{1,2}$}]
    {\small $EqAnt_{1\: \LBI}$} and {\small $EqAnt_{2\:\LBI}$} are admissible in {\small $[\LBIN + EqAnt_{1, 2\:\LBI}]$}, preserving
  the derivation depth.
   \label{admissible_eqant_LBI3}
\end{proposition}
\begin{proof}    
  Follows from  inversion lemma,\footnote{
  Inversion lemma proves one direction.}
  Proposition \ref{admissible_weakening_LBI3} 
  and Proposition \ref{admissible_EA2}.
  \qed 
  \hide{
  We divide {\small $EqAnt_{1\: \LBI}$} into two rules, one for each direction: 
   \begin{center}
     {\small 
     \AxiomC{$\Gamma(\Delta; \emptyset) \vdash F$} 
     \RightLabel{$EqAnt'_1$} 
     \UnaryInfC{$\Gamma(\Delta) \vdash F$} 
     \DisplayProof 
     \indent 
     \AxiomC{$\Gamma(\Delta) \vdash F$} 
     \RightLabel{$EqAnt''_1$} 
     \UnaryInfC{$\Gamma(\Delta; \emptyset) \vdash F$} 
     \DisplayProof
     }
   \end{center}  
   Proposition \ref{admissible_weakening_LBI3} for $EqAnt''_1$
   and the inversion lemma (Lemma \ref{LBITH_inversion_lemma}) 
   for the other two. \qed
   \hide{
   We only show a sketch here: 
   \begin{enumerate}
     \item {\small $EqAnt''_1$} is admissible due to Proposition 
       \ref{admissible_weakening_LBI3}. 
     \item {\small $EqAnt'_1$}: introduction (in backward derivation) of the additive 
       structural unit 
       {\small $\emptyset$} is required if there is a possibility 
       that, upon an application of an additive
       {\LBITH} inference rule 
       on a sequent {\small $\Gamma(\Delta) \vdash F$} in which 
       {\small $\Delta$}
       is not empty, the antecedent part of at least one 
       of the premise sequents becomes empty. No 
       {\LBITH} additive inference rules allow it; 
       note the formula duplicate to appear on the left 
       premise in {\small $\rightarrow L$}.
     \item {\small $EqAnt'_2$}: introduction of the multiplicative structural unit
       {\small $\O$} is required if there is a possibility that, 
       upon an application of a multiplicative {\LBITH} inference rule  
       on a sequent {\small $\Gamma \vdash F$}, 
       $\Gamma$ is distributed entirely onto one of 
       the premise sequents and a zero resource
       {\small ``$\O$''} must be distributed to  
       the other premise sequent.
       Such a
       situation arises for {\small $* R$} and {\small $\text{\wand} L$}.
       However, the possibility is already taken into account in {\small 
       $* R_2$} and resp.
       {\small $\text{\wand} L_{3,4}$} due to  
       $EqAnt_{2\: \LBI}$ which is internalised
       within them.
   \end{enumerate} 
   } 
   }
\end{proof} 
\subsection{Preparation for contraction admissibility in 
{\small $* R$/$\text{\wand} L$} cases}
We dedicate one subsection here to prepare for the main 
proof of contraction admissibility. Based on  
Proposition \ref{admissible_weakening_LBI3}, 
we make an observation concerning the set of candidates. 
The discovery, which is to be stated in 
Proposition \ref{lemma_lemma}, led to the solution to the open problem. 
\begin{definition}[Representing candidates] 
    Let {\small $\dpreceq: \mathfrak{S} \times 
        \mathfrak{S}$} be a reflexive and 
    transitive binary relation satisfying:
 \begin{multicols}{2}
     \begin{itemize} 
         \item {\small $\Gamma_1 \dpreceq \Gamma_2$} 
             if {\small $\Gamma_1 = \Gamma_2$}.   
         \item {\small $\Gamma_1 \dpreceq\: \Gamma_1; \Gamma_3$}. 
         \item {\small $[\Gamma_1 \dpreceq \Gamma_2] 
                 \wedge^{\dagger} [\Gamma_2 \dpreceq 
                 \Gamma_3] \rightarrow^{\dagger} 
                 [\Gamma_1 \dpreceq \Gamma_3]$}.   
         \item {\small $\Gamma_1, \Gamma_2\: \dpreceq\
                 \Gamma_1, (\Gamma_2; \Gamma_3)$}. 
     \end{itemize}  
 \end{multicols}  
 Now let {\small $\Gamma$} be a {\BI} structure. Then 
 any of the following pairs is a representing candidate 
 of {\small $\Gamma$}. 
 \begin{multicols}{2} 
     \begin{itemize}
         \item {\small $(\Gamma_x, {\mtop})$} if 
             {\small $\Gamma_x \dpreceq \Gamma$}. 
         \item {\small $(\Gamma_x, \Gamma_y)$} if 
             {\small $\Gamma_x, \Gamma_y \dpreceq \Gamma$}. 
     \end{itemize}
 \end{multicols} 
 We denote the set of representing candidates 
 of {\small $\Gamma$} by {\small $\code{RepCandidate}(\Gamma)$}. 
\end{definition}  
We trivially have that {\small $\code{RepCandidate}(\Gamma) \subseteq \code{Candidate}(\Gamma)$} for any {\small $\Gamma$}. More 
can be said. 
\begin{proposition}[Sufficiency of \code{RepCandidate}] 
    {\LBIN} with \code{RepCandidate} instead of 
    \code{Candidate} for {\small $(Re_1, Re_2)$} 
    is as expressive 
    as {\LBIN} (with \code{Candidate}). 
    \label{lemma_lemma} 
\end{proposition}  
\begin{proof}  
    The only inference rules in {\LBIN} that use 
    \code{Candidate} are {\small $* R$} and 
    {\small ${\text{\wand}} L$}. So it suffices to 
    consider only those. \\
    \indent For {\small $* R$}, suppose by way of showing contradiction that {\LBIN} with \code{RepCandidate} is not 
    as expressive as \LBIN, then there exists some {\LBIN} 
    derivation tree {\small $\Pi(D)$}: 
    \begin{center} 
        \scalebox{0.9}{  
            \AxiomC{$\vdots$} 
            \noLine
            \UnaryInfC{$D_1: Re_i \vdash F_1$}  
            \AxiomC{$\vdots$}
            \noLine 
            \UnaryInfC{$D_2: Re_j \vdash F_2$} 
            \RightLabel{$* R$} 
            \BinaryInfC{$D: \Gamma' \vdash F_1 * F_2$} 
            \DisplayProof
        }
    \end{center}   
    such that {\small $(Re_1, Re_2)$} must be 
    in {\small $\code{Candidiate}(\Gamma')\backslash 
        \code{RepCandidate}(\Gamma')$}. Now, 
    without loss of generality assume 
    {\small $(i, j) = (1, 2)$}. 
    Then  {\small $D_1': Re_i' \vdash F_1$} 
    and {\small $D_2': Re_j' \vdash F_2$} 
    for {\small $(Re_i', Re_j') \in \code{RepCandidate}(\Gamma')$} 
    are also {\LBIN} derivable (by Proposition \ref{admissible_weakening_LBI3}). But 
    this means that we can choose the {\small $(Re_i', Re_j')$} 
    for {\small $(Re_1, Re_2)$}, a direct contradiction to the supposition. Similarly for {\small ${\text{\wand}} L$}. \qed
\end{proof} 
  Contraction admissibility in {\LBIN} follows.
  \begin{theorem}[Contraction admissibility in \LBIN]
    If {\small $D: \Gamma(\Gamma_a; \Gamma_a) \vdash F$} is {\LBIN}-derivable, 
    then so is {\small $D': \Gamma(\Gamma_a) \vdash F$}, preserving 
  the derivation depth. 
  \label{admissible_contraction_LBI3}
\end{theorem} 
\begin{proof}     
By induction on the derivation depth of {\small $D$}.  
For an interesting case, we have 
{\small $* R$}. {\small $\Pi(D)$} looks like: 
       \begin{center} 
	 {\small 
	 \AxiomC{$\vdots$}
	 \noLine
	 \UnaryInfC{$D_1: Re_i \vdash F_1$}  
	 \AxiomC{$\vdots$}
	 \noLine
	 \UnaryInfC{$D_2: Re_j \vdash F_2$} 
	 \RightLabel{$* R$} 
	 \BinaryInfC{$D: \Gamma(\Gamma_a; \Gamma_a) \vdash F_1 * F_2$} 
	 \DisplayProof 
	 }
       \end{center} 
       By Proposition \ref{lemma_lemma}, assume that {\small $(Re_1, Re_2) \in \code{RepCandidate}(\Gamma(\Gamma_a; \Gamma_a))$}
       without 
       loss of generality. Then by the definition of 
       {\small $\dpreceq$} it must be that 
       either  (1) 
       {\small $\Gamma_a; \Gamma_a$} preserves completely 
       in {\small $Re_1$} or {\small $Re_2$}, or 
       (2) it remains neither in {\small $Re_1$} nor in 
       {\small $Re_2$}. 
       If {\small $\Gamma_a; 
       \Gamma_a$}
   is preserved in {\small $Re_1$} (or {\small $Re_2$}), then induction hypothesis on the premise that has 
   {\small $Re_1$} (or {\small $Re_2$}) and then 
         {\small $* R$} 
        conclude; otherwise, it is trivial to see that only 
       a single {\small $\Gamma_a$} needs to be present in {\small $D$}.  
Details are in Appendix C.
\end{proof}
\hide{
\begin{lemma}[Necessary implicit weakening through $* R_{\LBIT'}$]   
  In an application of $* R_{\LBIT'}$ on a derivable \LBIT' sequent 
  $D: \Gamma \vdash F * G$, such that  
  $\Gamma \equiv \mathcal{A}^{top} \equiv 
  F_1; \dots; F_m; \mathcal{M}_1; \dots; \mathcal{M}_n$, 
  implicit weakening on $\mathcal{A}^{\top}$ must take place 
 
  \label{}
\end{lemma}
\begin{lemma}[Weakening order]
  Given a derivation in \LBIT': 
  \begin{center}
     \Axioms{$D_1: Re_1 \vdash F$} 
     \Axioms{$D_2: Re_2 \vdash G$} 
     \RightLabel{$* R_{\LBIT}$} 
     \BinaryInfC{$D: \Delta \vdash F * G$} 
     \DisplayProof
   \end{center}  
   if $D$ is derivable, and $D_1$ and $D_2$ 
   are both derivable,  
    
\end{lemma} 
}    
\subsection{Equivalence of {\LBIN} to \LBI} 
\begin{theorem}[Equivalence between {\LBIN} and {\LBI}] 
  {\small $D: \Gamma \vdash F$} is \LBIN-derivable if and only if 
  it is \LBI-derivable.
  \label{equivalence_LBI3_LBI}
\end{theorem} 
\begin{proof}   
Into the \emph{only if} direction, assume that {\small $D$} 
  is \LBIN-derivable, and then show that there is 
  a \LBI-derivation for each {\LBIN} derivation. But this is 
  obvious because each {\LBIN} inference rule is derivable 
  in \LBI.\footnote{Note that {\small $EA_2$} 
      is \LBI-derivable with {\small $Wk L_{\LBI}$} and {\small $EqAnt_{2\:\LBI}$}.} \\
  \indent Into the \emph{if} direction, assume that {\small $D$}
  is \LBI-derivable, and then show that there is a corresponding 
  \LBIN-derivation to each {\LBI} derivation 
  by induction on the derivation depth of 
  {\small $D$}. Details are in Appendix D.
 \end{proof}  
\subsection{{\LBIN} Cut Elimination} 
Cut is admissible in [{\LBIN} + \Cut]. 
As a reminder (although already stated under 
Figure 1) {\Cut} is the following rule: 
\begin{center}
  \AxiomC{$\Gamma_1 \vdash F$} 
  \AxiomC{$\Gamma_2(F) \vdash G$} 
  \RightLabel{\Cut} 
  \BinaryInfC{$\Gamma_2(\Gamma_1) \vdash G$} 
  \DisplayProof
\end{center}
Just as in the case of intuitionistic 
logic, cut admissibility proof for a contraction-free
{\BI} sequent calculus is simpler than that for {\LBI} 
\cite{Arisaka2012-4}. Since we have already proved 
depth-preserving weakening admissibility, the following 
context sharing cut, $\Cut_{CS}$, is easily verified derivable in {\LBIN} + \Cut: 
{\small
\begin{center}
    \AxiomC{$\widetilde{\Gamma_3}; \Gamma_1 \vdash F$} 
  \AxiomC{$\Gamma_2(F; \Gamma_1) \vdash H$} 
  \RightLabel{$\Cut_{CS}$} 
  \BinaryInfC{$\Gamma_2(\widetilde{\Gamma_3}; \Gamma_1) \vdash H$} 
  \DisplayProof
\end{center}
} 
\noindent  where {\small $\Gamma_1$} appears on both of the premises. 
{\small $F$} in the above cut instance which is upward introduced 
on both premises is called the cut formula (for the cut instance).
The use of $\Cut_{CS}$ is just because 
it simplifies the cut elimination proof. \\
\indent For the proof, we recall the standard notations of 
cut rank and cut level.   
\begin{definition}[Cut level/rank]
  Given a cut instance in a closed derivation:{\ }\\ 
  \begin{center} 
      \scalebox{0.9}{ 
    \AxiomC{$D_1: \Gamma_1 \vdash F$} 
    \AxiomC{$D_2: \Gamma_2(F) \vdash H$} 
    \RightLabel{$\Cut$} 
    \BinaryInfC{$D_3: \Gamma_2(\Gamma_1) \vdash H$} 
    \DisplayProof 
}
  \end{center}
  \noindent The level of the cut instance is: {\small $\derivationdepth(D_1) + 
  \derivationdepth(D_2)$}, where {\small $\derivationdepth(D)$} denotes 
  derivation depth of {\small $D$}. The rank of the cut instance 
  is the size of the cut formula {\small $F$}, {\small $\formulasize(F)$},
  which is defined as follows: 
  \begin{itemize}
    \item it is 1 if {\small $F$} is a nullary logical 
        connective or a propositional variable. 
    \item it is {\small $\formulasize(F_1) + \formulasize(F_2) + 1$} 
      if {\small $F$} is in the form: {\small $F_1 \bullet F_2$} 
      for {\small $\bullet \in \{\wedge, \vee, \supset, 
      *, \text{\wand}\}$}.
  \end{itemize}
\end{definition}

\begin{theorem}[Cut admissibility in \LBIN]
  {\Cut} is admissible within {\LBIN} + \Cut.
  \label{cut_elimination}
\end{theorem}
\begin{proof} 
By induction on the cut rank and a sub-induction 
  on the cut level, by making use of $\Cut_{CS}$. 
  Details are in Appendix E.
     \qed
\end{proof} 
\hide{
\section{Towards Truly Contraction-Free BI Sequent Calculus - A Preliminary Investigation}   
In this section, an improvement over {\LBITH} is pondered upon 
based on a well-known implicit contraction elimination 
in the standard intuitionistic logic \cite{Dyckhoff92}.
Though {\LBITH} is contraction-free 
in the sense that an explicit structural rule of 
contraction does not appear within, 
the termination property of the calculus is not immediately 
apparent. This is due to, though 
only implicit, contraction that is internalised in the two implications.
We pave a way towards the complete {\BI} decidability analysis within
sequent calculus by demonstrating an extension of 
the said technique to multiplicative connectives 
{\small ``$\top^*$''} and {\small ``$*$''}. 
The following sub-calculus $\LBITH_1$ is specifically examined.
\begin{definition}[$\LBITH_1$] $\LBITH_1$
  comprises the following inference rules: 
  \begin{description}
    \item[Axioms: ] {\small $\quad\qquad id_{\LBITH} \quad 
      \top^* R_{\LBITH} \quad \bot L_{\LBITH} 
      \quad \top R_{\LBITH}$} 
    \item[Other logical rules: ] {\small $\: \qquad\qquad\qquad\quad 
      \wedge L_{\LBITH} \quad \wedge R_{\LBITH} \quad
      \vee L_{\LBITH} \quad \vee R_{\LBITH}$\\
      $\quad
      \rightarrow L_{\LBITH} \quad \rightarrow R_{\LBITH}
      \quad * L_{\LBITH} \quad * R_{1,2\: \LBITH} \quad 
      \top L \quad \top^* L$}
    \item[Structural rule: ] {\small $\qquad\quad\qquad\quad 
      EqAnt''_{2\: \LBITH} \quad \O L_{\LBITH}$} \\
  \end{description}
\end{definition}  
\hide{
\begin{definition}[$\LBITH_2$] $\LBITH_2$ 
  comprises the following inference rules: 
  \begin{description}
    \item[Axioms: ] $\quad id_{\LBITH} \quad 
      \top^* R_{\LBITH} \quad \bot L_{\LBITH} 
      \quad \top R_{\LBITH}$
    \item[Other logical rules: ]$\: \quad\qquad\qquad\quad 
      \wedge L_{\LBITH} \quad \wedge R_{\LBITH} \quad 
      \vee L_{\LBITH} \quad \vee R_{\LBITH}$\\
      $\quad * L_{\LBITH} \quad * R_{1,2\: \LBITH} 
      \quad \text{\wand} L_{1,2,3,4\: \LBITH} \quad 
      \text{\wand} R_{\LBITH}$ 
    \item[Structural rules: ] $\quad\qquad\quad 
      EqAnt''_{2\: \LBITH} \quad \O L_{\LBITH}$\\
  \end{description}
\end{definition}    
}
In accordance with the restriction, 
we assume the availability of only those connectives 
valid in $\LBITH_1$ to all the sequents appearing 
in a $\LBITH_1$ derivation.
\subsection{Contraction elimination in $\LBITH_1$}     
Our intention here is to prove that 
replacement of 
{\small $\rightarrow L_{\LBITH_1}$} with those in Figure 
\ref{LBI4_calculus} is sound and complete, to render $\LBITH_1$ 
nearly almost contraction-free even implicitly.\\
\subsubsection{Preparation}{\ }\\
We define three concepts for the main proofs: (1) 
non-theorem formulas; (2)
the derivation length; and (3) irreducible 
$\LBITH_1$ sequents.  
\begin{definition}[Non-theorem formulas]
  We call a $\LBITH_1$ formula $F$ a non-theorem 
  formula if it satisfies {\small $\O \not\vdash F$}, and 
  define it by {\small $F^{\star}$} to mean such.\\
\end{definition}
\begin{definition}[Derivation length]
  Given a derivation of a sequent {\small $D$},
  denoted as {\small $\Pi(D)$}, the derivation length 
  between {\small $D$} and some sequent {\small $D'$} that occurs 
  in {\small $\Pi(D)$} 
  $(${\small $\derivationlength(D, D')$}$)$ is inductively 
  defined as follows: 
  \begin{itemize}
    \item {\small $\derivationlength(D, D') = 0$} if 
      {\small $D'$} coincides with {\small $D$}. 
    \item {\small $\derivationlength(D, D') = 
      \derivationlength(D'', D') + 1$} if there exists 
      a sequent transition from {\small $D$} into {\small $D''$} and from 
      {\small $D''$} into {\small $D'$} such that {\small $D \leadsto D'' 
      \leadsto^* D'$}.\\
  \end{itemize}
\end{definition} 
\begin{definition}[Irreducible $\LBITH_1$ sequents]{\ }\\
  An antecedent structure {\small $\Gamma$} in $\LBITH_1$
  is said to be irreducible 
  if it contains as its substructure none of the following:  
  (1) {\small $p; p {\rightarrow} G$}
  (2) {\small $\O; \top^* {\rightarrow} G$}
  (3) {\small $\top {\rightarrow} G$}
  (4) {\small $\O; \O$} (5) {\small $\Delta, \O$} (6) {\small $\bot$} (7) {\small $\top$} 
  (8) {\small $\top^*$}
  (9) {\small $H_1 \wedge H_2$} (10) {\small $H_1 \vee H_2$} (11) {\small 
  $H_1 * H_2$}.
  A $\LBITH_1$ sequent {\small $D: \Gamma \vdash F$} is said 
  to be irreducible if {\small $\Gamma$} is irreducible. 
  \label{some_some_definition}
\end{definition}  
Notice that a $\LBITH_1$ sequent {\small $D$} that contains 
any one of the 11 structures can be reduced via the inversion lemma (Lemma 
\ref{LBITH_inversion_lemma}) such that {\small $D$} is $\LBITH_1$-derivable
if and only if the reduced 
sequent(s) is (are). \\
\subsubsection{A subformula 
property in {\small $\supset L_p$}, 
{\small $\supset L_{\top^*}$} and {\small $\supset L_{* 1}$}}{\ }\\
We now show that any
{\small $\supset L_{\LBITH_1}$} application on {\small ``$p {\supset} G$''} 
or {\small ``$\top^* {\supset} G$''} can be deferred 
until {\small ``$p;p {\supset} G$''} or resp.
{\small ``$\O; \top^* {\supset} G$''} appears on the antecedent 
part. Such also 
is the case 
for {\small $(F_1^{\star} * F_2^{\star}) {\supset} F_3$} under 
a set of conditions.
\begin{lemma}
  Any $\LBITH_1$-derivable irreducible sequent {\small $D: \Gamma 
  \vdash H$} has a closed derivation in which 
  the principal of the last rule applied is neither 
  {\small $p {\supset} G$} nor {\small $\top^* {\supset} G$}
  (on the 
  antecedent part of {\small $D$}), nor {\small $(F^{\star}_1 * F^{\star}_2) {\supset} G$} 
  if not all of the following conditions satisfy: 
\begin{itemize}
    \item {\small $[D: \Gamma(\Delta'; (F_1^{\star} * 
      F_2^{\star}) {\supset} G) \vdash 
      H] \leadsto_{\supset L}$\\
      $[D_1: 
      \Delta'; (F_1^{\star} * F_2^{\star}) {\supset} G
      \vdash F_1^{\star} * F_2^{\star}]$}
    \item {\small $D_1 \leadsto_{* R_1} [D_2: 
      Re_1 \vdash F_1^{\star}]$}
    \item {\small $D_1 \leadsto_{* R_1} [D_3: 
      Re_2 \vdash F_2^{\star}]$} 
    \item {\small $D_2$} and {\small $D_3$} (and hence also 
      {\small $D_1$}) are both $\LBITH_1$-derivable.
  \end{itemize}   
  \label{lemma_observation}
\end{lemma} 
\begin{proof}
   By contradiction. As in \cite{Dyckhoff92}, we assume without 
   loss of generality that inference rules to apply 
   in the leftmost branch were cleverly chosen such that
   the derivation length between {\small $D$} and 
   the conclusion sequent of an axiom in the leftmost branch
   be shortest.
   Suppose, by way
   of showing contradiction, that there cannot exist 
   any other shorter derivations of {\small $D$} than the ones ending in 
   {\small $\supset L$} for which a formula 
   in the form either {\small $p {\supset} G$},
   {\small $\top^* {\supset} G$} or {\small $(F^{\star}_1 * 
   F^{\star}_2) {\supset} G$} that 
   does not satisfy all the four given conditions 
   is its principal, 
   then 
   {\small $\Pi(D)$}, a derivation of $D$, looks like: 
   \begin{center}
     {\small 
     \AxiomC{$\vdots$} 
     \noLine  
     \UnaryInfC{$D_3$}   
     \AxiomC{$\vdots$}
     \noLine
     \UnaryInfC{$D_4$} 
     \RightLabel{\textbf{Inf}}
     \BinaryInfC{$D_1: \Delta; B {\supset} G \vdash B$} 
     \AxiomC{$\vdots$} 
     \noLine
     \UnaryInfC{$D_2: \Gamma(\Delta; G) \vdash H$} 
     \RightLabel{$\supset L$} 
     \BinaryInfC{$D: \Gamma(\Delta; B {\supset} G) 
     \vdash H$} 
     \DisplayProof
     }
   \end{center} 
   where {\small $B$} is {\small $p$} if the principal is 
   {\small ``$p {\supset} G$''}; is {\small $\top^*$} if 
   {\small ``$\top^* 
   {\supset} G$''}; 
   or is {\small ``$F_1^{\star} * F_2^{\star}$''} if {\small ``$(F_1^{\star} 
   * F_2^{\star}) {\supset} G$''}.  As {\small $D$} 
   is irreducible, so is {\small $D_1$} which, therefore, 
   cannot be the conclusion sequent of an axiom.  
   If {\small $B$} is either {\small $p$} or {\small $\top^*$}, then 
   the consequent formula of {\small $D_1$} can be active 
   only for an axiom. Likewise, due to the given 
   set of conditions, if {\small $B$} is {\small $F^{\star}_1 * F^{\star}_2$} 
   in {\small $D_1$}, its consequent part cannot be active 
   for \textbf{Inf}. Therefore, \textbf{Inf} is known to be
   {\small $\supset L$}.
   Moreover, as 
   the leftmost branch is supposed shortest, the principal 
   for \textbf{Inf} must be from among those constituents 
   residing in the same additive structural layer as 
   the {\small $B {\supset} G$}. Furthermore, that the leftmost 
   branch is shortest has to dictate that the principal 
   for \textbf{Inf} is  in neither of the following:  
   {\small ``$p_i {\rightarrow} G_i$''}, {\small ``$\top^*_j {\rightarrow} G_j$''} 
   or {\small ``$(F^{\star}_{k1} * F^{\star}_{k2}) {\rightarrow}
   G_k$''} 
   for some propositional variable {\small $p_i$}, some 
   multiplicative logical unit {\small $\top^*_j$}, 
   some {\small $F^{\star}_{k1} * F^{\star}_{k2}$} (satisfying 
   the same set of the conditions as stated) and 
   some formula {\small $G_i$}, {\small $G_j$} or {\small $G_k$}. \\
   \indent These points taken into account,  
   {\small $D$, $D_1$, $D_2$, $D_3$} and {\small $D_4$} are actually 
   seen taking the following forms for some other formula {\small $F$}:
   \begin{itemize}
     \item {\small $D: \Gamma(\Delta'; F {\rightarrow} G'; B {\rightarrow} G)
       \vdash H$}
     \item {\small $D_1: \Delta'; F {\rightarrow} G'; B {\rightarrow}
       G \vdash B$}
     \item {\small $D_2: \Gamma(\Delta'; F {\rightarrow} G'; G) 
       \vdash H$}
     \item {\small $D_3: \Delta'; F {\rightarrow} G'; B {\rightarrow}
       G \vdash F$} 
     \item {\small $D_4: \Delta'; G'; B {\rightarrow} G \vdash B$}
   \end{itemize}
   But, then, this perforce implies the existence 
   of an alternative derivation {\small $\Pi'(D)$}  which 
   results by permuting
   {\small $\Pi(D)$}: 
   \begin{center}    
     {\small 
     \AxiomC{$\vdots$}
     \noLine
     \UnaryInfC{$D_3$} 
     \AxiomC{$\vdots$}
     \noLine
     \UnaryInfC{$D_4$}
     \AxiomC{$\vdots$}
     \noLine
     \UnaryInfC{$D_2': \Gamma(\Delta'; G'; G) \vdash H$}
     \RightLabel{$\rightarrow L$}
     \BinaryInfC{$\Gamma(\Delta'; G'; B {\rightarrow} G)
     \vdash H$}
     \RightLabel{$\rightarrow L$}
     \BinaryInfC{$D$}
     \DisplayProof 
     }
   \end{center}
   {\small $D'_2$} can be shown derivable  
   from {\small $D_2$} via the inversion lemma (Lemma \ref{LBITH_inversion_lemma}). A direct contradiction to 
   the supposition has been drawn, for the leftmost 
   branch in {\small $\Pi'(D)$} is shorter.
\end{proof} 
{\ }\\      
From Lemma \ref{lemma_observation} follows an
observation.
\begin{lemma}  
  In $\LBITH_1$, {\small $\rightarrow L_{*}'$} as below is admissible 
  in {\small $\rightarrow L_{*1}$}. \\
  \begin{center}  
    \scalebox{0.9}{
     \AxiomC{$D_1: \Delta; (F_1^{\star} * F_2^{\star}) 
     {\rightarrow} G \vdash F_1^{\star} * F_2^{\star}$} 
     \AxiomC{$D_2: \Gamma(\Delta; G) \vdash H$} 
     \RightLabel{$\rightarrow L_{*}'$} 
     \BinaryInfC{$D: \Gamma(\Delta; (F_1^{\star} * 
     F_2^{\star}) {\rightarrow} G) \vdash H$} 
     \DisplayProof  
     }
  \end{center} 
  \label{last_proposition}
\end{lemma}
\begin{proof} 
  Any application of {\small $\rightarrow L$} with 
  {\small $(F_1^{\star} * F_2^{\star}) {\rightarrow} G$} as its principal 
can be deferred until all the four conditions hoisted in 
Lemma \ref{lemma_observation} are satisfied. Under 
the assumption, there exists a pair of sequent transitions 
via {\small $* R_1$}
from the left premise sequent {\small $D_1$} of {\small $\rightarrow L_{*}'$} into {\small $D_2$} and {\small $D_3$} such that (1) {\small 
$D_1 \leadsto_{* R_1}  
D_2\:$} ; (2) {\small $D_1 \leadsto_{* R_1} D_3\:$} ; and (3) both {\small $D_2$} and {\small $D_3$}
are $\LBITH_1$-derivable. Then, because the antecedent part 
of {\small $D_1$} is not a multiplicative structural layer (checked 
by eye inspection on {\small $\rightarrow L_{*}'$}) nor can it be
a formula {\small $(F_1^{\star} * F_2^{\star}) 
{\rightarrow} G$} (otherwise 
{\small $D_1$} is not $\LBITH_1$-derivable), 
it must be an additive structural layer, and moreover, 
it must be such that there exists at least one multiplicative 
structural layer as its constituent (because 
the four conditions in Lemma \ref{lemma_observation} are assumed satisfied). By the process 
of a maximal {\small $Re_1$/$Re_2$} generation (\emph{c.f.} Lemma 
\ref{lemma_lemma} 
and Corollary \ref{maximal_pair}),
it cannot be 
the case that two constituents 
of the outermost additive structural layer be retained 
simultaneously.
And so there could be only one from among the $\mathcal{M}$ 
constituents which is to remain after a sequence of {\small $Wk L_{1}$}
(\emph{c.f.} Lemma \ref{lemma_lemma}) 
so that the result be a multiplicative 
structural layer to appear at the outermost structural layer.
But {\small $(F_1 * F_2) {\rightarrow} F_3$} is not a multiplicative 
structural layer.
\end{proof}    
\begin{proposition}
  Replacement of {\small $\rightarrow L_{\LBITH_1}$} with 
  those in Figure \ref{LBI4_calculus} is sound and complete.
\end{proposition} 
\begin{proof}
   One direction: to assume inference rules 
   in Figure \ref{LBI4_calculus} and to show 
   corresponding derivations with {\small $\rightarrow L_{\LBITH_1}$},
   is trivial. Into the other direction, 
   we consider what the actual instance $F$ is 
   in the principal {\small $F {\rightarrow} G$}, and 
   turn to Lemma \ref{lemma_observation} and Lemma \ref{last_proposition}, 
   for 
   {\small $\rightarrow L_{* 1}$}, {\small $\rightarrow L_{p}$} and 
   {\small $\rightarrow L_{\top^*}$}. {\small $\rightarrow L_{\top}$}
   is straightforward. For the other cases, Dyckhoff92 \cite{Dyckhoff92}.
\end{proof} 
\vspace{0.2cm}
\subsubsection{Remaining issues}{\ }\\ 
We have extended \cite{Dyckhoff92} to multiplicative 
connectives in the previous subsection, rendering $\LBITH_1$
nearly almost contraction-free even implicitly.
Nevertheless, there still is a certain challenge in establishing 
a subformula property for {\small $\rightarrow L_{*2}$}. 
The difficulty we see is as follows\footnote{A correction to an error 
in the submitted
paper here to make the example effective according to the intention 
of the authors.}:
\begin{enumerate} 
  \item Suppose a partial derivation tree comprising the following 
    sequents:
    \begin{itemize} 
      \item
    {\small $[D: \Gamma(\Delta; p {\rightarrow} (F_4 * F_5)
    ; (F_1 * (p {\rightarrow} F_2)) {\rightarrow} 
    G) \vdash H]$} 
  \item {\small $[D_1: \Delta; p {\rightarrow} (F_4 * F_5); 
    (F_1 * (p {\rightarrow} F_2)) {\rightarrow} 
    G) \vdash F_1 * (p {\rightarrow} F_2)]$}  
  \item {\small $[D_2: \Gamma(\Delta; p {\rightarrow} (F_4 * F_5)
    ; G) \vdash H]$}
  \item {\small $[D_3: \O \vdash F_1]$}
  \item 
    {\small $[D_4: \Delta; p {\rightarrow} (F_4 * F_5)
    ; (F_1 * (p {\rightarrow} F_2)) {\rightarrow} 
    G \vdash p \rightarrow F_2]$}
  \item 
    {\small $[D'_4: \Delta; p; p{\rightarrow} (F_4 * F_5); 
    (F_1 * (p {\rightarrow} F_2)) {\rightarrow} G \vdash F_2]$} 
  \item 
    {\small $[D''_4: \Delta; p; F_4 * F_5; 
    (F_1 * (p {\rightarrow} F_2)) {\rightarrow} G \vdash F_2]$} 
  \item 
    {\small $[D'''_4: \Delta; p; (F_4, F_5); 
    (F_1 * (p {\rightarrow} F_2)) {\rightarrow} G \vdash F_2]$}  
  \item
    {\small $[D_5: \Delta; p; (F_4, F_5);
    (F_1 * (p {\rightarrow} F_2)) {\rightarrow} G \vdash 
    F_1 {*} (p {\rightarrow} F_2)]$}  
  \item 
    {\small $[D_6: \Delta; p; (F_4, F_5); 
    G \vdash F_2]$}
  \item   
    {\small $[D_7: F_4 \vdash F_1]$} 
  \item 
    {\small $[D_8: F_5 \vdash p {\rightarrow} F_2]$}
\end{itemize} 
such that 
\begin{center} 
  \hspace{-0.5cm}
  \scalebox{0.85}{ 
   \AxiomC{}
   \doubleLine
   \UnaryInfC{$D_3$}  
   \AxiomC{$\vdots$}
   \noLine
   \UnaryInfC{$D_7$}  
   \AxiomC{$\vdots$}
   \noLine
   \UnaryInfC{$D_8$}
   \RightLabel{$* R_1$}
   \BinaryInfC{$D_5$}  
   \AxiomC{$\vdots$}
   \noLine
   \UnaryInfC{$D_6$}
   \RightLabel{$\rightarrow L_{*2}$}
   \BinaryInfC{$D'''_4$}
   \dottedLine
   \RightLabel{Inversion}
   \UnaryInfC{$D''_4$}
   \dottedLine
   \RightLabel{Inversion}
   \UnaryInfC{$D'_4$} 
   \dottedLine
   \RightLabel{Inversion}
   \UnaryInfC{$D_4$} 
   \RightLabel{$* R_2$}
   \BinaryInfC{$D_1$}   
   \AxiomC{$\vdots$}
   \noLine
   \UnaryInfC{$D_2$}
   \RightLabel{$\rightarrow L_{*2}$}
   \BinaryInfC{$D$}
   \DisplayProof 
   }
\end{center}    
\item In the above derivation, the formula 
  {\small ``$(F_1 * (p {\rightarrow} F_2)) \rightarrow G$''} becomes 
  the principal for {\small $\rightarrow L_{*2}$} twice without 
  a redundancy in the given partial derivation tree.
\item If it were eliminated on {\small $D_1$}, then 
  {\small $\rightarrow L_{*2}$} 
  could not apply on {\small $D'''_4$} for the same formula,
  and could then affect the derivability of {\small $D$}.
\end{enumerate}  
Indeed, the problem is more immanent than just 
within {\small $\rightarrow L_{* 2}$}: although 
{\small $\text{\wand} L_{* 1}$} works as 
intended, it may not be known whether some formula {\small $F$} 
is a non-theorem formula. 
The last analysis step towards the syntactical conclusion of {\BI} 
decidability analysis should first concentrate on gaining 
a finer understanding of the proof-theoretical behaviour of 
$\LBITH_1$. By solving the simpler problem for $\LBITH_1$,  
an extension of the analysis framework to {\LBITH} may well be in scope.
}\section{Conclusion}            
We solved an open problem of structural rule absorption 
in {\BI} sequent calculus. This problem stood unsolved for a 
while. As far back as we can see, the first attempt was made 
in \cite{OHearn03}. References to the problem 
were subsequently made \cite{DonnellyGKMP04,Brotherston10-4,journals/mscs/GalmicheMJP05}. The work 
that came closest to ours is one by Donnelly {\it et al.} \cite{DonnellyGKMP04}. 
They consider weakening absorption in the context of 
forward theorem proving (where weakening than contraction 
is a source of non-termination). 
One inconvenience in their approach, however, is that 
the effect of weakening is not totally isolated from that of contraction: 
it is absorbed into contraction as well as into logical rules.
But then structural weakening is still possible through 
the new structural contraction. Also, 
the coupling of the two structural rules amplifies 
the difficulty of analysis on the behaviour of contraction. 
Further, their work is on a subset of {\BI} without units. In 
comparison, our solution covers the whole \BI.   
And our analysis fully decoupled the effect of structural weakening 
from the effect of structural contraction. {\LBIN} comes 
with no structural rules, in fact. Techniques we used in this work 
should be useful for deriving a contraction-free 
sequent calculus of other non-classical logics coming with 
a non-formula contraction. There are also more recent 
{\BI} extensions in sequent calculus such as \cite{Kamide13-2}, to which this work 
has relevance. 
\bibliographystyle{plain}
\bibliography{references} 
\clearpage  
\section*{Appendix A: Proof of Proposition \ref{admissible_weakening_LBI3}} 
By induction on  derivation depth of {\small $D$}. 
  If it is one, \emph{i.e.} 
  {\small $D$} is the conclusion sequent of an axiom, then so is {\small $D'$}. 
  For inductive cases, assume that the current proposition holds for  
  all the derivations of depth up to $k$.
  It must be now demonstrated 
  that it still holds for derivations of depth $k +1$.
  Consider what the last inference rule is in {\small $\Pi(D)$}.
  \begin{enumerate}  
    \item {\small $\supset L$:  $\Pi(D)$} looks like:
      \begin{center}  
	{\small 
        \AxiomC{$\vdots$}
	\noLine
    \UnaryInfC{$\mathbb{E}(\widetilde{\Gamma_1}; F {\supset} G) \vdash F$}   
	 \AxiomC{$\vdots$} 
	 \noLine 
     \UnaryInfC{$\Gamma(G; \mathbb{E}(\widetilde{\Gamma_1}; F {\supset} 
	 G)) \vdash H$} 
	 \RightLabel{$\supset L$} 
     \BinaryInfC{$\Gamma(\mathbb{E}(\widetilde{\Gamma_1}; F {\supset} G))
	 \vdash H$} 
	 \DisplayProof
	 }
      \end{center} 
      By induction hypothesis on both of the premises, 
      {\small $\mathbb{E}'(\widetilde{\Gamma_1'};  F {\supset} G)
      \vdash F$} and\linebreak
  {\small $\Gamma'(G; \mathbb{E}'(\widetilde{\Gamma_1'}; F {\supset}
      G)) \vdash H$} are 
      both \LBIN-derivable. Here we assume that:\\
      {\small $
          \mathbb{E}'(\widetilde{\Gamma_1'}; F \supset G) \preceq \mathbb{E}(\widetilde{\Gamma_1'}; F \supset G) \preceq \mathbb{E}(\widetilde{\Gamma_1}; F \supset G)$}, and 
      {\small $\Gamma' (-) \preceq \Gamma(-)$}. \\
      Then  
      {\small $\Gamma'(\mathbb{E}'(\widetilde{\Gamma_1'}; 
      F \supset G))
      \vdash H$} is also \LBIN-derivable via 
      {\small $\supset L$}.  
    \item {\small $\text{\wand} L$: $\Pi(D)$} looks like:
      \begin{center} 
	\scalebox{0.9}{ 
	\AxiomC{$\vdots$} 
	\noLine
	 \UnaryInfC{$Re_i \vdash F$}   
	 \AxiomC{$\vdots$}
	 \noLine
     \UnaryInfC{$\Gamma((\widetilde{Re_j}, G); 
         (\widetilde{\Gamma_1}, \mathbb{E}(\widetilde{\Gamma_2}; F \text{\wand} G)))
	  \vdash H$} 
	  \RightLabel{$\text{\wand} L$} 
      \BinaryInfC{$\Gamma(\widetilde{\Gamma_1}, \mathbb{E}(\widetilde{\Gamma_2}; F \text{\wand} 
	  G)) 
	 \vdash H$} 
	 \DisplayProof 
	 }
      \end{center}  
      Assume that {\small $\widetilde{\Gamma_1'} 
          \preceq \widetilde{\Gamma_1}$} and that  
      {\small $\mathbb{E}'(\widetilde{\Gamma_2'}; F {\text{\wand}}
          G) \preceq \mathbb{E}(\widetilde{\Gamma_2'}; F {\text{\wand}}
          G) \preceq \mathbb{E}(\widetilde{\Gamma_2}; F {\text{\wand}} G)$}. Then by induction hypothesis on the right 
      premise 
      sequent,\\ 
      {\small $\Gamma'((\widetilde{Re_j}, G); (\widetilde{\Gamma_1'},
          \mathbb{E}'(\widetilde{\Gamma_2'}; F {\text{\wand}} G)))
          \vdash H$} is \LBIN-derivable. \\
      Then {\small $\Gamma'(\widetilde{\Gamma_1'}, 
          \mathbb{E}'(\widetilde{\Gamma_2'}; F {\text{\wand}} G))
          \vdash H$} is also \LBIN-derivable via
      {\small ${\text{\wand}} L$}. 
    \item Other cases are simpler and similar.  \qed
  \end{enumerate}   
  \section*{Appendix B: Proof of Lemma \ref{LBITH_inversion_lemma}}  
By induction on the derivation depth $k$.  
We abbreviate {\small $(\Gamma(\Gamma_1))(\Gamma_2)$} 
by {\small $\Gamma(\Gamma_1)(\Gamma_2)$}. And we also 
do not explicitly show a tilde on top of a possibly empty 
structure. 
  \begin{enumerate}
    \item For a {\LBIN } sequent {\small $\Gamma(F \wedge G) \vdash H$}, 
       the base case is when it is an axiom, and the proof is trivial.
       For inductive cases, assume that the 
       statement holds true for all the derivation depths 
       up to $k$, and show that it still holds true at 
       $k + 1$. Consider what the last 
       inference rule applied is.  
       \begin{enumerate}
	 	 \item $\vee L$: 
             The derivation ends in: {\ }\\
	   {\small
	   \begin{center}
	     \AxiomC{$\Gamma(F \wedge G)(F_1) \vdash H$}
	     \AxiomC{$\Gamma(F \wedge G)(F_2) \vdash H$} 
	     \RightLabel{$\vee L$}
	     \BinaryInfC{$\Gamma(F \wedge G)(F_1 \vee F_2) \vdash H$}
	     \DisplayProof
	   \end{center} 
	   }{\ }\\
           By induction hypothesis, both 
	   {\small $\Gamma(F; G)(F_1) \vdash H$} and 
	   {\small $\Gamma(F; G)(F_2) \vdash H$} are \LBIN-derivable. 
	   Then {\small $\Gamma(F; G)(F_1 \vee F_2) \vdash H$} as required via 
	   {\small $\vee L$}.  
\item $\wedge L$: Similar, or trivial when the principal 
	   should coincide with {\small $F \wedge G$}.  
	 \item $\supset L$: The derivation 
	   ends in one of the following:{\ }\\ 
	   \begin{center}
	     \scalebox{0.9}{ 
	     \AxiomC{$\mathbb{E}(\Gamma_1(F \wedge G); F_1 \supset G_1) \vdash F_1$}
	     \AxiomC{$\Gamma(G_1; \mathbb{E}(\Gamma_1(F\wedge G); F_1 {\supset} G_1)) \vdash H$}
	      \RightLabel{$\supset L$} 
	      \BinaryInfC{$\Gamma(\mathbb{E}(\Gamma_1(F \wedge G); 
	      F_1 \supset G_1)) \vdash H$}
	      \DisplayProof 
	      }
	      {\ }\\{\ }\\{\ }\\
	      \scalebox{0.9}{
	      \AxiomC{$\mathbb{E}(\Gamma'_1; 
	      F_1 \supset G_1) \vdash F_1$}
	      \AxiomC{$\Gamma'(F \wedge G)(G_1; \mathbb{E}(\Gamma'_1; F_1 {\supset} G_1)) \vdash H$}
	      \RightLabel{$\supset L$}
	      \BinaryInfC{$\Gamma'(F \wedge G)(\mathbb{E}(\Gamma'_1; 
	      F_1 \supset
	      G_1)) \vdash H$}
	      \DisplayProof 
	      }
	   \end{center} 
	   {\ }\\
	   By induction hypothesis, both {\small $\mathbb{E}(\Gamma_1(F; G);
	   F_1 \supset G_1) \vdash 
	   F_1$} and\\ {\small $\Gamma(G_1; \mathbb{E}(\Gamma_1(F; G); F_1
	   {\supset} G_1))
	   \vdash H$} 
	   in case 
	   the former, or \\
	   {\small $\Gamma'(F; G)(G_1; 
	   \mathbb{E}(\Gamma'_1; F_1 {\supset} G_1)) \vdash H$} 
	   in case the latter.\\ Then {\small $\supset L$} (with the untouched
	   left premise if the latter) produces the required result. 
	 \item $* L$: The derivation ends in: {\ }
	   {\small 
	   \begin{center}
	     \AxiomC{$\Gamma(F \wedge G)(F_1, G_1) \vdash H$}
	     \RightLabel{$* L$}
	     \UnaryInfC{$\Gamma(F \wedge G)(F_1 * G_1) \vdash H$}
	     \DisplayProof
	   \end{center} 
	   }{\ }\\
	   By induction hypothesis, {\small 
	   $\Gamma(F; G)(F_1, G_1) \vdash H$}. 
	   Then, {\small $\Gamma(F; G)(F_1 * G_1) \vdash H$} as required via 
	   {\small $* L$}.  
	 \item $\text{\wand} L$: The derivation
	   ends in one of the following, depending on 
	   the location at which {\small $F \wedge G$} appears. 
	   In the below inference steps, we assume that 
	   the particular formula {\small $F \wedge G$} 
       occurs in {\small $Re_{(i, j)}(F \wedge G)$} as 
       the focused substructure, but not in {\small $Re_{(i, j)}$}.\footnote{
	   Note, however, that this does not preclude 
	   occurrences of {\small $F \wedge G$} in case 
	   it occurs multiple times in the conclusion sequent.}
	   \\
	   \begin{center}  
            \scalebox{0.9}{ 
	    \AxiomC{$Re_i \vdash F_1$} 
	    \AxiomC{$\Gamma((Re_j, G_1); (\Gamma', \mathbb{E}(
	    \Gamma_1(F \wedge G); F_1 \text{\wand} G_1)))
	    \vdash H$} 
	    \BinaryInfC{$\Gamma((\Gamma', \mathbb{E}(\Gamma_1(F \wedge G); 
	    F_1 \text{\wand} G_1))) \vdash H$} 
	    \DisplayProof 
	    }{\ }\\{\ }\\{\ }\\
            \scalebox{0.9}{  
             \AxiomC{$Re_i \vdash F_1$} 
	     \AxiomC{$\Gamma((Re_j, G_1); 
	     (\Gamma'(F \wedge G), \mathbb{E}(\Gamma_1;F_1 \text{\wand} 
	     G_1))) \vdash H$} 
	     \BinaryInfC{$\Gamma((\Gamma'(F \wedge G),
	     \mathbb{E}(\Gamma_1; F_1 \text{\wand} G_1))) \vdash H$} 
             \DisplayProof
	     }{\ }\\{\ }\\{\ }\\ 
	     \scalebox{0.9}{ 
	     \AxiomC{$Re_i(F \wedge G) \vdash F_1$} 
	     \AxiomC{$\Gamma((Re_j, G_1); (\Gamma'(F \wedge G), 
	     \mathbb{E}(
	     \Gamma_1; F_1 \text{\wand} G_1))) \vdash H$} 
	     \BinaryInfC{$\Gamma((\Gamma'(F \wedge G), 
	     \mathbb{E}(\Gamma_1; F_1 \text{\wand} G_1))) \vdash H$} 
	     \DisplayProof
	     }{\ }\\{\ }\\{\ }\\ 
	     \scalebox{0.9}{
	     \AxiomC{$Re_i \vdash F_1$}
	     \AxiomC{$\Gamma((R_j(F \wedge G), G_1); 
	     (\Gamma'(F \wedge G), \mathbb{E}(\Gamma_1; F_1 \text{\wand} 
	     G_1))) \vdash H$} 
	     \BinaryInfC{$\Gamma((\Gamma'(F \wedge G),
	     \mathbb{E}(\Gamma_1; F_1 \text{\wand} G_1))) \vdash H$} 
	     \DisplayProof
	     }{\ }\\{\ }\\{\ }\\ 
	     \scalebox{0.9}{ 
	     \AxiomC{$Re_i \vdash F_1$} 
	     \AxiomC{$\Gamma(F \wedge G)((Re_j, G_1); 
	     (\Gamma', \mathbb{E}(\Gamma_1; F_1 \text{\wand} G_1)))
	     \vdash H$} 
	     \BinaryInfC{$\Gamma(F \wedge G)((\Gamma', 
	     \mathbb{E}(\Gamma_1; F_1 \text{\wand} G_1))) \vdash H$} 
	     \DisplayProof
	     }{\ }\\{\ }\\{\ }\\ 
	     \scalebox{0.9}{ 
	     \AxiomC{$Re_i \vdash F_1$}   
	     \AxiomC{$\Gamma((Re_j, G_1); \mathbb{E}(\Gamma_2(F \wedge G)
	     (\Gamma', (\Gamma_1; F_1 {\text{\wand}} G_1)))) 
	     \vdash H$} 
	     \BinaryInfC{$\Gamma(\mathbb{E}(\Gamma_2(F \wedge G)(\Gamma',
	     (\Gamma_1; F_1 {\text{\wand}} G_1)))) \vdash H$} 
	     \DisplayProof
	     }
          \end{center} 
	  {\ }\\
	   For each, the required sequent results from 
	   induction hypothesis for the particular occurrences 
	   of {\small $F \wedge G$}
	   on both of the premises, and then 
	   {\small $\text{\wand} L$}
	   to recover 
	   {\small $\Gamma'$} (or {\small $\Gamma'(F; G)$})
       such that  {\small $(Re_i, Re_j) \in \code{Candidate}(
           \Gamma') \text{ (or } \code{Candidate}(\Gamma'(F; G)))$}. 
	 \item $\wedge R$: Similar to {\small $\vee L$} in approach but simpler.
	 \item $\vee R$: Similar.
	 \item $\supset R$: The derivation ends in: {\ }\\
	   {\small 
	   \begin{center}
	     \AxiomC{$\Gamma(F \wedge G); F_1 \vdash G_1$}
	     \RightLabel{$\supset R$}
	     \UnaryInfC{$\Gamma(F \wedge G) \vdash F_1 \supset G_1$}
	     \DisplayProof
	   \end{center} 
	   }{\ }\\
	   By induction hypothesis, {\small $\Gamma(F; G); F_1 \vdash G_1$}.
	   Then, {\small $\Gamma(F; G) \vdash F_1 \supset G_1$}
	   as required 
	   via {\small $\supset R$}.  
	 \item $* R$: The derivation 
	   ends in one of the below: {\ }\\
	   {\small 
	   \begin{center} 
	     \AxiomC{$Re_i \vdash F_1$} 
	     \AxiomC{$Re_j \vdash G_1$} 
	     \BinaryInfC{$\Gamma'(F \wedge G) \vdash F_1 * G_1$} 
	     \DisplayProof 
	     {\ }\\{\ }\\ {\ }\\
	     \AxiomC{$Re_i(F \wedge G) \vdash F_1$} 
	     \AxiomC{$Re_j \vdash G_1$} 
	     \BinaryInfC{$\Gamma'(F \wedge G) \vdash F_1 * G_1$}
	     \DisplayProof
	     {\ }\\{\ }\\{\ }\\  
\AxiomC{$Re_i \vdash F_1$} 
\AxiomC{$Re_j(F \wedge G) \vdash G_1$} 
	     \BinaryInfC{$\Gamma'(F \wedge G) \vdash F_1 * G_1$} 
	     \DisplayProof 
	   \end{center}   
	   }
	   Trivial for the first case. For the second, 
	   induction hypothesis on the left premise sequent 
	   produces {\small $Re_i(F; G) \vdash F_1$}.\\ Then 
	   {\small $* R$} such that 
       {\small $(Re_i(F; G), Re_j) \in \code{Candidate}(
           \Gamma'(F; G))$}. Similarly 
       for the third case. 
	 \item $\text{\wand} R$: Trivial.
       \end{enumerate}  
     \item A {\LBIN} sequent {\small $\Gamma(F \vee G) \vdash H$}:
       similar.  
     \item For a {\LBIN } sequent {\small $\Gamma(F * G) \vdash H$}, 
       the base case is when it is an axiom for which a 
       proof is trivially given. 
	For inductive cases, assume that it holds true for 
	all the derivation depths up to $k$ and show that the same 
	still holds for the derivation depth of $k+1$. Consider 
	what the last inference rule is.   
	\begin{enumerate} 
	    \hide{
	  \item $\rightarrow L$: 
	    One of the following:{\ }\\
	    \begin{center}  
	      \scalebox{0.9}{
	      \AxiomC{$\Gamma_1(F * G); F_1 \rightarrow G_1 \vdash F_1$}
	      \AxiomC{$\Gamma(\Gamma_1(F* G); G_1) \vdash H$}
	      \RightLabel{$\rightarrow L$}
	      \BinaryInfC{$\Gamma(\Gamma_1(F * G); F_1 \rightarrow G_1)
	      \vdash H$}
	      \DisplayProof  
	      }  
	      {\ }\\{\ }\\{\ }\\
	      \scalebox{0.9}{ 
	      \AxiomC{$\Gamma_1; F_1 \rightarrow G_1 \vdash H$}
	      \AxiomC{$\Gamma(F * G) (\Gamma_1; G_1) \vdash H$}
	      \RightLabel{$\rightarrow L$}
	      \BinaryInfC{$\Gamma(F*G)(\Gamma_1; F_1 \rightarrow G_1) 
	      \vdash H$}
	      \DisplayProof   
	      }{\ }\\{\ }\\
	    \end{center} 
	    For each, induction hypothesis on the premise(s) 
	    for each occurrence of {\small $F * G$}, and then 
	    {\small $\rightarrow L$} to conclude. 
	    }
	  \item $* L$: Trivial if the principal coincides with 
	    {\small $F * G$}. Otherwise, the derivation looks like: 
	    {\small 
	    \begin{center}
	      \AxiomC{$\Gamma(F * G)(F_1, G_1) \vdash H$}
	      \RightLabel{$* L$}
	      \UnaryInfC{$\Gamma(F * G)(F_1 * G_1) \vdash H$}
	      \DisplayProof
	    \end{center} 
	    } 
	    {\ }\\
	    By induction hypothesis, {\small 
	    $\Gamma(F, G)(F_1, G_1) \vdash H$}. 
	    Then, {\small $\Gamma(F, G)(F_1 * G_1) \vdash H$} as desired 
	    via {\small $* L$}.  
	  \item The rest: Similar to the previous cases.
	\end{enumerate}  
      \item 
	For a {\LBIN} sequent 
	{\small $D: \Gamma(\Gamma_1, \mtop) \vdash H$}, 
	the base case is when it is the conclusion sequent of an axiom. 
\begin{enumerate}
  \item $id$: {\small $D: \mathbb{E}(\Gamma'(\Gamma_1, \mtop); p) \vdash p$}. 
     Then
     {\small $D': \mathbb{E}(\Gamma'(\Gamma_1); p) \vdash p$} 
    is also an axiom.  
  \item $\abot L$, {\small $\top R$}: straightforward.  
  \item $\mtop R$: similar to $id$ case. 
\end{enumerate}
For inductive cases, assume that the statement holds true 
for all the derivation depths up to $k$, and show that 
it still holds true at $k+1$. Consider what the last inference rule
applied is. 
\begin{enumerate}
  \item $\vee L$: The derivation ends in one of the following:{\ }\\ 
    {\small 
    \begin{center}
      \AxiomC{$\Gamma(\Gamma_1, \mtop)(F_1) \vdash H$} 
      \AxiomC{$\Gamma(\Gamma_1, \mtop)(F_2) \vdash H$} 
      \RightLabel{$\vee L$} 
      \BinaryInfC{$\Gamma(\Gamma_1, \mtop)(F_1 \vee F_2) \vdash H$} 
      \DisplayProof 
      {\ }\\{\ }\\{\ }\\
      \AxiomC{$\Gamma(\Gamma_1(F_1), \mtop) \vdash H$} 
      \AxiomC{$\Gamma(\Gamma_1(F_2), \mtop) \vdash H$} 
      \RightLabel{$\vee L$} 
      \BinaryInfC{$\Gamma(\Gamma_1(F_1 \vee F_2), \mtop) \vdash H$} 
      \DisplayProof
    \end{center}  
    }
    {\ }\\
    For the former, {\small $\Gamma(\Gamma_1)(F_1) \vdash H$}
    and {\small $\Gamma(\Gamma_1)(F_2) \vdash H$} (induction hypothesis); 
    then {\small $\Gamma(\Gamma_1)(F_1 \vee F_2) \vdash H$} via 
    {\small $\vee L$} as required. For the latter,
    {\small $\Gamma(\Gamma_1(F_1)) \vdash H$} and 
    {\small $\Gamma(\Gamma_1(F_2)) \vdash H$} (induction hypothesis); 
    then {\small $\Gamma(\Gamma_1(F_1 \vee F_2)) \vdash H$} via 
    {\small $\vee L$} as required. 
  \item $\supset L$: The derivation ends in one of the following:
    {\ }\\
    \begin{center} 
      {\small
      \AxiomC{$\mathbb{E}(\Gamma_1; F_1 {\supset} F_2)
       \vdash F_1$} 
       \AxiomC{$\Gamma(F_2; \mathbb{E}(\Gamma_1; F_1 {\supset} F_2))(\Gamma_2, \mtop) \vdash H$} 
       \RightLabel{$\supset L$} 
       \BinaryInfC{$\Gamma(\mathbb{E}(\Gamma_1; F_1 {\supset} F_2))
       (\Gamma_2, \mtop)
       \vdash H$} 
       \DisplayProof  
       } 
       {\ }\\{\ }\\{\ }\\
       {\small 
       \AxiomC{$\mathbb{E}(\Gamma_1(\Gamma_2, \mtop); F_1 {\supset} 
       F_2) \vdash F_1$} 
       \AxiomC{$\Gamma(F_2; \mathbb{E}(\Gamma_1(\Gamma_2, \mtop); F_1 {\supset} F_2))
       \vdash H$} 
       \RightLabel{$\supset L$} 
       \BinaryInfC{$\Gamma(\mathbb{E}(\Gamma_1(\Gamma_2, \mtop); F_1 {\supset} F_2))
       \vdash H$} 
       \DisplayProof 
       } 
       {\ }\\{\ }\\{\ }\\
       {\small 
       \AxiomC{$\mathbb{E}(\Gamma_1; F_1 {\supset} F_2) \vdash F_1$} 
       \AxiomC{$\Gamma(\Gamma_2(F_2; \mathbb{E}(\Gamma_1; F_1 {\supset} F_2)), \mtop) \vdash H$} 
       \RightLabel{$\supset L$} 
       \BinaryInfC{$\Gamma(\Gamma_2(\mathbb{E}(\Gamma_1; F_1 {\supset} F_2)), 
       \mtop) \vdash H$} 
       \DisplayProof 
       }
    \end{center} 
    {\ }\\
    For the first, 
    {\small $\Gamma(F_2; \mathbb{E}(\Gamma_1; F_1 {\supset} F_2))(\Gamma_2) \vdash H$}
    (induction hypothesis); then\linebreak 
    {\small $\Gamma(\mathbb{E}(\Gamma_1; F_1 {\supset} F_2))(\Gamma_2)
    \vdash H$ via $\supset L$} as required. \\
    \indent For the second, {\small $\mathbb{E}(\Gamma_1(\Gamma_2); F_1 {\supset} 
   F_2) \vdash F_1$} and 
   {\small $\Gamma(\mathbb{E}(\Gamma_1(\Gamma_2); F_1 {\supset} F_2)) \vdash H$} 
    (induction hypothesis); then 
    {\small $\Gamma(\mathbb{E}(\Gamma_1(\Gamma_2); F_1 {\supset} F_2)) \vdash H$} 
    via {\small $\supset L$} as required. \\
    \indent For the third, induction hypothesis on the right 
    premise sequent, then {\small $\supset L$} to conclude. 
  \item $\text{\wand} L$: 
    Suppose 
    the derivation ends in one of the following:{\ }\\
    {\ }\\
    \begin{center}  
      \scalebox{0.9}{ 
      \AxiomC{$Re_i \vdash F$} 
      \AxiomC{$\Gamma((Re_j, G); (\Gamma_2, \mathbb{E}(\Gamma_3(\Gamma_1, \mtop);
      F {\text{\wand}} G))) \vdash H$} 
      \RightLabel{$\text{\wand} L$} 
      \BinaryInfC{$\Gamma((\Gamma_2, \mathbb{E}(\Gamma_3(\Gamma_1, \mtop); 
      F {\text{\wand}} G))) \vdash H$} 
      \DisplayProof  
      }
      {\ }\\{\ }\\{\ }\\  
      \scalebox{0.9}{ 
      \AxiomC{$Re_i \vdash F_1$} 
      \AxiomC{$\Gamma(\Gamma_1, \mtop)((Re_j, G); 
      (\Gamma_2, \mathbb{E}(\Gamma_3; F {\text{\wand}}
      G))) \vdash H$} 
      \RightLabel{$\text{\wand} L$} 
      \BinaryInfC{$\Gamma(\Gamma_1, \mtop)(\Gamma_2, \mathbb{E}(\Gamma_3; 
      F {\text{\wand}} G)) \vdash H$} 
      \DisplayProof 
      } 
      {\ }\\{\ }\\{\ }\\
      \scalebox{0.9}{ 
      \AxiomC{$Re_i \vdash F_1$} 
      \AxiomC{$\Gamma(\Gamma_1((Re_j, G); (\Gamma_2, \mathbb{E}(\Gamma_3; 
      F {\text{\wand}} G)), \mtop)) \vdash H$} 
      \RightLabel{$\text{\wand} L$} 
      \BinaryInfC{$\Gamma(\Gamma_1((\Gamma_2, \mathbb{E}(\Gamma_3; 
      F {\text{\wand}} G)), \mtop)) \vdash H$} 
      \DisplayProof 
      }
    \end{center}
    {\ }\\
    For each of the above, induction hypothesis, if applicable, and {\small 
    $\text{\wand} L$}
    conclude. Now consider  other cases where 
    the {\small $\mtop$} occurs 
    in the conclusion sequent as
    {\small $\Gamma(\mathbb{E}(\Gamma_2(\Gamma_1, \mtop), (\Gamma_3; F {\text{\wand}} G)))
    \vdash H$}. Less involved cases 
    are when\linebreak
    {\small ``$\Gamma_1, \mtop$''} is entirely retained 
    or entirely discarded upwards:{\ }\\
    \begin{center}
      \scalebox{0.9}{ 
      \AxiomC{$Re_i(\Gamma_1, \mtop) \vdash F$} 
      \AxiomC{$\Gamma((Re_j, G); (\Gamma_2(\Gamma_1, \mtop), \mathbb{E}(\Gamma_3; 
      F {\text{\wand}} G))) \vdash H$} 
      \RightLabel{$\text{\wand} L$} 
      \BinaryInfC{$\Gamma((\Gamma_2(\Gamma_1, \mtop),
      \mathbb{E}(\Gamma_3; F {\text{\wand}} G))) \vdash H$} 
      \DisplayProof  
      }
      {\ }\\{\ }\\{\ }\\  
      \scalebox{0.9}{
      \AxiomC{$Re_i \vdash F$} 
      \AxiomC{$\Gamma((Re_j(\Gamma_1, \mtop), G); (\Gamma_2(\Gamma_1, 
      \mtop), \mathbb{E}(\Gamma_3; F {\text{\wand}} G))) \vdash H$} 
      \RightLabel{$\text{\wand} L$} 
      \BinaryInfC{$\Gamma((\Gamma_2(\Gamma_1, \mtop), 
      \mathbb{E}(\Gamma_3; F {\text{\wand}} G))) \vdash H$} 
      \DisplayProof  
      }
      {\ }\\{\ }\\{\ }\\  
      \scalebox{0.9}{
      \AxiomC{$Re_i \vdash F$} 
      \AxiomC{$\Gamma((Re_j, G); (\Gamma_2(\Gamma_1, \mtop), 
      \mathbb{E}(\Gamma_3; F {\text{\wand}} G))) \vdash H$} 
      \RightLabel{$\text{\wand} L$} 
      \BinaryInfC{$\Gamma((\Gamma_2(\Gamma_1, \mtop), \mathbb{E}(\Gamma_3; 
      F {\text{\wand}} G))) \vdash H$} 
      \DisplayProof  
      }
    \end{center} 
    {\ }\\
    The first assumes that the specific {\small ``$\Gamma_1, \mtop$"}
    does not occur
    in {\small $Re_j$}; the second that it does not occur in {\small 
    $Re_i$}; 
    the third that it does not occur in {\small $Re_i$} or in {\small 
    $Re_j$}.  
    Each of them is concluded via induction hypothesis 
    and then {\small $\text{\wand} L$}. \\
    \indent Finally, if {\small ``$\Gamma_1, \mtop$''} should be 
    split between the two premises, {\small $Re_j$} is the 
    {\small $\mtop$}, in which case we have 
    on the right premise sequent:\\
    {\small $\Gamma((\mtop, G); (\Gamma_2(\Gamma_1, {\mtop}), 
        \mathbb{E}(\Gamma_3; F {\text{\wand}} G))) \vdash H$}. \\
    In this case we apply induction hypothesis and obtain\\
    {\small $\Gamma((\mtop, G); (\Gamma_2(\Gamma_1), 
        \mathbb{E}(\Gamma_3; F {\text{\wand}} G))) \vdash H$}.\\ 
    By the definition of a candidate, however, we have from the 
    sequent that\\ 
    {\small $\Gamma(\Gamma_2(\Gamma_1), 
        \mathbb{E}(\Gamma_3; F {\text{\wand}} G)) \vdash H$} \\
    is \LBIN-derivable, as required. 
  \item The rest: similar or straightforward.
\end{enumerate}	 
     \item The rest: similar or straightforward.   
  \end{enumerate} 
  \section*{Appendix C: Proof of Theorem \ref{admissible_contraction_LBI3}}\label{appendix_contraction} 
  By induction on derivation depth. 
  The base cases 
  are when it is 1, \emph{i.e.} when 
  {\small $D$} is the conclusion sequent of an axiom. Consider which 
  axiom has applied. If it is {\small $\top R$}, then it is trivial 
  to show that if {\small $\Gamma(\Gamma_a; \Gamma_a) 
  \vdash \top$}, then so is {\small $\Gamma(\Gamma_a) \vdash \top$}. Also for 
  {\small $\abot L$}, a single occurrence of {\small $\abot$} on the antecedent part of 
  {\small $D$} suffices for the {\small $\abot L$} application, and the current 
  theorem is trivially provable in this case, too. For both {\small $id$} and 
  {\small $\mtop R$}, 
  {\small $\Pi(D)$} looks like: 
  \begin{center} 
    {\small 
    \AxiomC{}
    \UnaryInfC{$\mathbb{E}(\widetilde{\Gamma_1}; \alpha) \vdash \alpha$} 
    \DisplayProof 
    }
  \end{center} 
  where {\small $\alpha$} is {\small $p \in \mathcal{P}$} for {\small $id$}, {\small $\mtop$} for 
  {\small $\mtop R$} and 
  {\small $\Gamma(\Gamma_a; \Gamma_a) = \mathbb{E}(\widetilde{\Gamma_1}; \alpha)$}.
   If 
  {\small $\alpha$} is not a sub-structure of  
  either of the occurrences of {\small $\Gamma_a$}, 
  then {\small $D'$} is trivially derivable. 
  Otherwise, assume that the focused {\small $\alpha$} 
  in {\small $\mathbb{E}(\widetilde{\Gamma_1}; \alpha)$} 
  is a sub-structure of one of the occurrences of 
  {\small $\Gamma_a$} in {\small $\Gamma(\Gamma_a; \Gamma_a)$}.  
  Then there exists some  {\small $\Gamma_2$} and {\small $\widetilde{\Gamma_3}$}  such that 
  {\small $\mathbb{E}(\widetilde{\Gamma_1}; \alpha) = 
      \mathbb{E}(\Gamma_2; \widetilde{\Gamma_3}; \alpha) 
      = \mathbb{E}_1(\Gamma_2); \mathbb{E}_2(\widetilde{\Gamma_3}; 
      \alpha)$} and that 
  {\small $\Gamma_a$} is an essence of {\small $\widetilde{\Gamma_3};
      \alpha$}. But then {\small $D': \Gamma(\Gamma_a)$} is still 
  an axiom. \\
  \indent For inductive cases, suppose that the current theorem 
  holds true 
  for any derivation depth of up to $k$. We must 
  demonstrated that it still holds for the derivation depth of $k+1$. 
  Consider what the {\LBIN} inference rule applied last is, and, 
  in case of a left inference rule, consider where 
  the active structure {\small $\Gamma_b$} of the inference rule  is
  in {\small $\Gamma(\Gamma_a; \Gamma_a)$}. 
  \begin{enumerate}
    \item {\small $\wedge L$}, and {\small $\Gamma_b$} is {\small 
      $F_1 \wedge F_2$}: 
  if {\small $\Gamma_b$} does not appear in {\small $\Gamma_a$}, induction hypothesis 
      on the premise sequent concludes. Otherwise, 
      {\small $\Pi(D)$} looks like: 
      \begin{center} 
	{\small 
	\AxiomC{$\vdots$}
	\noLine
	 \UnaryInfC{$D_1: \Gamma(\Gamma_a'(F_1; F_2); 
	 \Gamma_a'(F_1 \wedge F_2)) \vdash H$} 
	 \RightLabel{$\wedge L$} 
	 \UnaryInfC{$D: \Gamma(\Gamma_a'(F_1 \wedge F_2); 
	 \Gamma_a'(F_1 \wedge F_2)) \vdash H$} 
	 \DisplayProof
	 }
      \end{center} 
      {\small $D'_1: \Gamma(\Gamma_a'(F_1; F_2); \Gamma_a'(F_1; F_2)) \vdash H$} 
      is \LBIN-derivable (inversion lemma); {\small $D''_1: 
      \Gamma(\Gamma_a'(F_1; F_2)) \vdash H$} is also \LBIN-derivable 
      (induction hypothesis); then {\small $\wedge L$} on {\small $D''_1$}
      concludes.   
     \item {\small $\supset L$}, and {\small $\Gamma_b$} is {\small $\mathbb{E}(\widetilde{\Gamma'}; F \supset G)$}: if 
         {\small $\Gamma_b$} does not appear in {\small $\Gamma_a$}, then the induction hypothesis 
       on both of the premises concludes. If it is entirely 
       in {\small $\Gamma_a$}, 
       then {\small $\Pi(D)$} looks either like:
       \begin{center}    
	 {\small
	 \AxiomC{$\vdots$} 
	 \noLine
     \UnaryInfC{$D_1: \mathbb{E}(\widetilde{\Gamma'};
	  F \supset G) \vdash F$}   
	  \AxiomC{$\vdots$}
	  \noLine
	  \UnaryInfC{$D_2$}	  
	  \RightLabel{$\supset L$} 
      \BinaryInfC{$D: \Gamma(\Gamma_a'(\mathbb{E}(\widetilde{\Gamma'}; F \supset G)); 
          \widetilde{\Gamma_a'}(\mathbb{E}(\widetilde{\Gamma'}; F \supset G))) \vdash H$} 
	  \DisplayProof  
	  }
       \end{center}   
       where {\small $D_2: \Gamma(\Gamma_a'(G; \mathbb{E}(\widetilde{\Gamma'}; F 
           {\supset} G)); \Gamma_a'(\mathbb{E}(\widetilde{\Gamma'};
       F {\supset} G))) \vdash H$},
       or, in case {\small $\Gamma_a$} is {\small $\Gamma_a'; F {\supset} G$}, like: 
       \begin{center}  
	 {\small 
	 \AxiomC{$\vdots$}
	 \noLine
	 \UnaryInfC{$D_1: \Gamma_a'; F {\supset} G; 
	 \Gamma_a'; F {\supset} G
	  \vdash F$}  
	  \AxiomC{$\vdots$} 
	  \noLine
	  \UnaryInfC{$D_2$}	  
	  \RightLabel{$\supset L$} 
	  \BinaryInfC{$D: \Gamma(\Gamma_a'; F {\supset} G;
	  \Gamma_a'; F {\supset} G) \vdash H$} 
	  \DisplayProof 
	  }
       \end{center}  
       where {\small $D_2: \Gamma(G; \Gamma_a'; F {\supset} G; \Gamma_a';
       F {\supset} G) \vdash H$}. \\
       In the former case,\\ {\small $D_2': \Gamma(\Gamma_a'(G; 
           \mathbb{E}(\widetilde{\Gamma'}; F {\supset} G));
           \Gamma_a'(G; \mathbb{E}(\widetilde{\Gamma'}; F {\supset} G))) \vdash 
       H$} (weakening admissibility);\\
   {\small $D''_2: \Gamma(\Gamma_a'(G; \mathbb{E}(\widetilde{\Gamma'}; 
       F {\supset} G))) \vdash H$} (induction hypothesis); \\
       then {\small $\supset L$} on {\small $D_1$} and {\small 
       $D''_2$} 
       concludes. In the latter, induction hypothesis on {\small $D_1$}
       and on {\small $D_2$}; then via
       {\small $\supset L$} for a conclusion. Finally, if 
       only a substructure of {\small $\Gamma_b$} is in {\small $\Gamma_a$} with 
       the rest spilling out of {\small $\Gamma_a$}, then if 
       the principal formula {\small $F \supset G$} does not occur in {\small $\Gamma_a$}, then straightforward; otherwise 
       similar 
       to the latter case. 
     \item {\small $* R$}: {\small $\Pi(D)$} looks like: 
       \begin{center} 
	 {\small 
	 \AxiomC{$\vdots$}
	 \noLine
	 \UnaryInfC{$D_1: Re_i \vdash F_1$}  
	 \AxiomC{$\vdots$}
	 \noLine
	 \UnaryInfC{$D_2: Re_j \vdash F_2$} 
	 \RightLabel{$* R$} 
	 \BinaryInfC{$D: \Gamma(\Gamma_a; \Gamma_a) \vdash F_1 * F_2$} 
	 \DisplayProof 
	 }
       \end{center} 
       By Proposition \ref{lemma_lemma}, assume that {\small $(Re_1, Re_2) \in \code{RepCandidate}(\Gamma(\Gamma_a; \Gamma_a))$}
       without 
       loss of generality. Then by the definition of 
       {\small $\dpreceq$} it must be that 
       either  (1) 
       {\small $\Gamma_a; \Gamma_a$} preserves completely 
       in {\small $Re_1$} or {\small $Re_2$}, or 
       (2) it remains neither in {\small $Re_1$} nor in 
       {\small $Re_2$}. 
       If {\small $\Gamma_a; 
       \Gamma_a$}
   is preserved in {\small $Re_1$} (or {\small $Re_2$}), then induction hypothesis on the premise that has 
   {\small $Re_1$} (or {\small $Re_2$}) and then 
         {\small $* R$} 
        conclude; otherwise, it is trivial to see that only 
       a single {\small $\Gamma_a$} needs to be present in {\small $D$}.  
     \item {\small $\text{\wand} L$}, and {\small $\Gamma_b$} is
         {\small $\widetilde{\Gamma'}, 
           \mathbb{E}(\widetilde{\Gamma_1}; F {\text{\wand}} G)$}:  
       if {\small $\Gamma_b$} is not in {\small $\Gamma_a$}, 
       then induction hypothesis 
       on the right premise sequent concludes. If it is in {\small $\Gamma_a$}, 
       {\small $\Pi(D)$} looks like: 
	 \begin{center}  
	   \scalebox{0.88}{  
	   \AxiomC{$\vdots$}
	    \noLine
	    \UnaryInfC{$D_1: Re_i \vdash F$}   
	    \AxiomC{$\vdots$}
	    \noLine 
	    \UnaryInfC{$D_2$}
	    \RightLabel{$\text{\wand} L_1$} 
	    \BinaryInfC{$D: 
            \Gamma(\Gamma_a'(\widetilde{\Gamma'}, \mathbb{E}(\widetilde{\Gamma_1}; {F \text{\wand} G})); 
            \Gamma_a'(\widetilde{\Gamma'}, \mathbb{E}(\widetilde{\Gamma_1}; {F \text{\wand} G}))) \vdash H$}
	    \DisplayProof 
	    }
	 \end{center}   
	 {\ }\\
	 where {\small $D_2$} is:
	 \begin{center} 
	   \scalebox{0.88}{
           $\Gamma(\Gamma_a'((\widetilde{Re_j}, G); (\widetilde{\Gamma'}, 
       \mathbb{E}(\widetilde{\Gamma_1}; F {\text{\wand}} G))); 
       \Gamma_a'(\widetilde{\Gamma'}, \mathbb{E}(\widetilde{\Gamma_1}; F {\text{\wand}} G))) \vdash H$   
    }
          \end{center}    
	  {\ }\\
      {\small $D'_2: \Gamma(\Gamma_a'((\widetilde{Re_j}, G); (
          \widetilde{\Gamma'}, 
          \mathbb{E}(\widetilde{\Gamma_1}; F {\text{\wand}} G)));
          \Gamma_a'((\widetilde{Re_j}, G);
          (\widetilde{\Gamma'}, \mathbb{E}(\widetilde{\Gamma_1}; F {\text{\wand}} G)))) \vdash H$} 
	  via Proposition \ref{admissible_weakening_LBI3} is 
      also \LBIN-derivable. {\small $D''_2: \Gamma(\Gamma_a'((\widetilde{Re_j}, G); 
          (\widetilde{\Gamma'}, \mathbb{E}(\widetilde{\Gamma_1}; F {\text{\wand}} G))))$
	  $\vdash H$} 
	  via induction hypothesis. Then 
	  {\small $\text{\wand} L$} on {\small $D_1$} and {\small 
	  $D''_2$ }
	  concludes. If, on the other hand, {\small $\Gamma_a$} is in 
	  {\small $\Gamma_b$}, 
      then it is  either in {\small $\Gamma_1$} or in 
	  {\small $\Gamma'$}. 
	  But if it is in {\small $\Gamma_1$}, 
	  then it must be weakened away, 
	  and if it is in {\small $\Gamma'$}, 
	  similar to the $* R$ case. 
	\item Other cases are similar to one of the cases already
            examined. 
  \end{enumerate}

  \section*{Appendix D: Proof of Theorem \ref{equivalence_LBI3_LBI}}\label{appendix_equivalence}   
Into the \emph{only if} direction, assume that {\small $D$} 
  is \LBIN-derivable, and then show that there is 
  a \LBI-derivation for each {\LBIN} derivation. But this is 
  obvious because each {\LBIN} inference rule is derivable 
  in \LBI.\footnote{Note that {\small $EA_2$} 
      is \LBI-derivable with {\small $Wk L_{\LBI}$} and {\small $EqAnt_{2\:\LBI}$}.} \\
  \indent Into the \emph{if} direction, assume that {\small $D$}
  is \LBI-derivable, and then show that there is a corresponding 
  \LBIN-derivation to each {\LBI} derivation 
  by induction on the derivation depth of 
  {\small $D$}. \\
  \indent If it is 1, \emph{i.e.} if {\small $D$} is the conclusion sequent 
  of an axiom, we note that {\small $\abot L_{\LBI}$} is identical to {\small 
  $\abot L_{\LBIN}$};
  {\small $id_{\LBI}$} and {\small $\mtop R_{\LBI}$} via 
  {\small $id_{\LBIN}$} and resp. {\small $\mtop R_{\LBIN}$} with Proposition 
  \ref{admissible_weakening_LBI3} and Proposition \ref{admissible_EA2}; and {\small $\top R_{\LBI}$} 
  is identical to {\small $\top R_{\LBIN}$}. For inductive cases, assume 
  that the \emph{if} direction holds true up to the {\LBI}-derivation 
  depth of $k$, then it must be demonstrated that it still holds true 
  for the {\LBI}-derivation depth of $k+1$. Consider what the {\LBI} 
  rule applied last is: 
  \begin{enumerate} 
    \item {\small $\supset L_{\LBI}$}: 
      {\small $\Pi_{\LBI}(D)$} looks like: 
      \begin{center}   
	{\small 
	\AxiomC{$\vdots$}
        \noLine
	\UnaryInfC{$D_1: \Gamma_1 \vdash F$}   
	\AxiomC{$\vdots$}
	\noLine
	\UnaryInfC{$D_2: \Gamma(\Gamma_1;G) \vdash H$} 
	\RightLabel{$\supset L_{\LBI}$} 
	\BinaryInfC{$D: \Gamma(\Gamma_1; F {\supset} G) \vdash H$} 
	\DisplayProof 
	}
      \end{center} 
      By induction hypothesis, both {\small $D_1$} and {\small $D_2$} are also 
      \LBIN-derivable. Proposition \ref{admissible_weakening_LBI3} 
      on {\small $D_1$} in \LBIN-space results in {\small $D'_1: 
      \Gamma_1; F {\supset} G \vdash F$}, 
      and on {\small $D_2$} results in {\small $D'_2: 
      \Gamma(\Gamma_1; G; F {\supset} G) \vdash H$}. Then 
      an application of {\small $\supset L_{\LBIN}$} on {\small
      $D'_1$} and {\small $D_2$}
      concludes in \LBIN-space. 
    \item {\small $\text{\wand} L_{\LBI}$}: 
      {\small $\Pi_{\LBI}(D)$} looks like: 
      \begin{center} 
	{\small
	 \AxiomC{$\vdots$} 
	 \noLine
	 \UnaryInfC{$D_1: \Gamma_1 \vdash F$} 
	 \AxiomC{$\vdots$} 
	 \noLine
	 \UnaryInfC{$D_2: \Gamma(G) \vdash H$} 
	 \RightLabel{$\text{\wand} L_{\LBI}$}  
	 \BinaryInfC{$D: \Gamma(\Gamma_1, F {\text{\wand}} G) \vdash H$} 
	 \DisplayProof 
	 }
      \end{center}  
      By induction hypothesis, {\small $D_1$} and {\small $D_2$} are also 
      \LBIN-derivable. 
      \begin{enumerate}
	\item If {\small $\Gamma(G)$} is {\small $G$}, \emph{i.e.} 
	  if the antecedent part of {\small $D_2$} is a formula 
	  ({\small $G$}), 
	  then Proposition \ref{admissible_weakening_LBI3} on {\small $D_2$}  
      results in {\small $D'_2: G; (\Gamma_1, F {\text{\wand}} G) \vdash H$}
	  in \LBIN-space. Then {\small $\text{\wand} L_{\LBIN}$} on 
	  {\small $D_1$} and {\small $D'_2$} leads to {\small $D': 
          \Gamma_1, F {\text{\wand}} G \vdash H$} as required.  
	\item If {\small $\Gamma(G)$} is {\small $\Gamma'(\Gamma'', G)$}, 
	  then Proposition \ref{admissible_weakening_LBI3} on {\small $D_2$} 
      leads to\linebreak {\small $D'_2: \Gamma'((\Gamma'', G); 
      (\Gamma'', \Gamma_1, 
      F {\text{\wand}} G)) \vdash H$}. Then {\small $\text{\wand} L_{ \LBIN}$}
      on {\small $D_1$} and {\small $D'_2$} leads to {\small $D': \Gamma'(\Gamma'', \Gamma_1, 
      F {\text{\wand}} G) \vdash H$} as required. 
    \item Finally, if {\small $\Gamma(G)$} is {\small $\Gamma'(\Gamma''; 
      G) \vdash H$}, then Proposition \ref{admissible_weakening_LBI3} 
      on {\small $D_2$} leads to\linebreak {\small $D'_2: \Gamma'(\Gamma''; G; (\Gamma_1, 
      F {\text{\wand}} G)) \vdash H$}. Then {\small $\text{\wand} 
      L_{\LBIN}$} on {\small $D_1$} and {\small $D'_2$} leads to {\small $D': 
      \Gamma'(\Gamma''; (\Gamma_1, F {\text{\wand}} G)) \vdash H$}  
      as required. 
  \end{enumerate} 
\item {\small $Wk L_{\LBI}$}: Proposition \ref{admissible_weakening_LBI3}.
\item {\small $Ctr L_{\LBI}$}: Theorem \ref{admissible_contraction_LBI3}.  
\item {\small $EqAnt_{1\: \LBI}$}: Proposition \ref{admissible_eqant_LBI3}. 
\item {\small $EqAnt_{2\: \LBI}$}: Proposition \ref{admissible_eqant_LBI3}. 
\item The rest: straightforward. 
  \end{enumerate}
  \section*{Appendix E: Proof of Theorem \ref{cut_elimination}}\label{appendix_cut_elimination} 
By induction on the cut rank and a sub-induction 
  on the cut level, by making use of $\Cut_{CS}$. 
  In this proof 
  $(X, Y)$ denotes, for some {\LBIN} inference rules $X$ and $Y$, 
  that one of the premises has been just derived with $X$ and 
  the other with $Y$. 
  As before, {\small $\Gamma(\Gamma_1)(\Gamma_2)$} 
  abbreviates {\small $(\Gamma(\Gamma_1))(\Gamma_2)$}.  
  In the pairs of derivations, the first is the derivation tree 
  to be permuted and the second is the permuted derivation tree. 
  \begin{description}
    \item[$(id, id)$: ]{\ }
      \begin{enumerate}
	\item{\ }  
	  {\small 
	  \begin{center}
	    \AxiomC{}
	    \RightLabel{$id$} 
        \UnaryInfC{$\mathbb{E}(\widetilde{\Gamma_1}; p) \vdash p$} 
	    \AxiomC{}
	    \RightLabel{$id$} 
        \UnaryInfC{$\mathbb{E}'(\widetilde{\Gamma_2}; p) \vdash p$} 
	    \RightLabel{$\Cut$} 
        \BinaryInfC{$\mathbb{E}'(\widetilde{\Gamma_2}; \mathbb{E}(\widetilde{\Gamma_1};
	    p)) \vdash p$} 
	    \DisplayProof
	  \end{center}  
	  }
	  {\ }\\
	  $\Rightarrow$
	  {\ }\\  
	  {\small
	  \begin{center}
	    \AxiomC{}
	    \RightLabel{$id$} 
        \UnaryInfC{$\mathbb{E}'(\widetilde{\Gamma_2}; \mathbb{E}(\widetilde{\Gamma_1}; p))
	    \vdash p$} 
	    \DisplayProof
	  \end{center}
	  }  
	  {\ }\\ 
	  Of course, for the above permutation to be correct, 
	  we must be able to demonstrate the fact that
	  the antecedent structure is 
      {\small $\mathbb{E}''(\widetilde{\Gamma_2}; \widetilde{\Gamma_1}; p)$} 
      such that\linebreak {\small $[\mathbb{E}''(\widetilde{\Gamma_2}; \widetilde{\Gamma_1}; p)] 
          = [\mathbb{E}'(\widetilde{\Gamma_2}; \mathbb{E}(
          \widetilde{\Gamma_1}; p))]$}. 
	  But note that it only takes a finite number of (backward)
	  $EA_2$ applications (\emph{Cf.} Proposition 
      \ref{admissible_EA2}) on {\small $\widetilde{\Gamma_2}; 
          \mathbb{E}(\widetilde{\Gamma_1}; p) \vdash p$} 
      to upward derive {\small $\widetilde{\Gamma_2}; \widetilde{\Gamma_1}; p \vdash p$}.
	  The implication is that, since  
      {\small $\widetilde{\Gamma_2}; \mathbb{E}(\widetilde{\Gamma_1}; p) \vdash p$}  
      results upward from {\small $\mathbb{E}'(\widetilde{\Gamma_2}; 
          \mathbb{E}(\widetilde{\Gamma_1}; p)) \vdash p$} also 
	  in a finite number of backward {\small $EA_2$} applications, 
	  the antecedent structure must 
	  be in the form:
      {\small $\mathbb{E}''(\widetilde{\Gamma_2}; \widetilde{\Gamma_1}; p)$}.   
	\item{\ }\\  
	  {\small 
	  \begin{center}
	    \AxiomC{}
	    \RightLabel{$id$} 
        \UnaryInfC{$\mathbb{E}(\widetilde{\Gamma_1}; p) \vdash p$} 
	    \AxiomC{}
	    \RightLabel{$id$} 
	    \UnaryInfC{$\mathbb{E}'(\Gamma_2(p); q) \vdash q$} 
	    \RightLabel{$\Cut$} 
        \BinaryInfC{$\mathbb{E}'(\Gamma_2(\mathbb{E}(\widetilde{\Gamma_1}; p)); 
	    q) \vdash q$} 
	    \DisplayProof
	  \end{center}
	  }
      \end{enumerate}
  \end{description} 
  Other patterns for which one of the premises is an axiom sequent 
  are straightforward. \\
  \indent For the rest, if the cut formula is principal  
  only for one of the premise sequents, then we follow the routine
  \cite{351148}
  to permute up the other premise sequent for which 
  it is the principal. For example, in case we 
  have the derivation pattern below:{\ }\\ 
  \begin{center} 
    \scalebox{0.77}{ 
    \hspace{-0.4cm}
    \AxiomC{$D_1$}
    \AxiomC{$D_2$}
    \RightLabel{$\vee L$} 
    \BinaryInfC{$D_5: \Gamma_1(H_1 \vee H_2) \vdash F_1 {\supset} F_2$} 
    \AxiomC{$D_3: \mathbb{E}(\widetilde{\Gamma_3}; F_1 {\supset} F_2) \vdash F_1$} 
    \AxiomC{$D_4: \Gamma_2(F_2; \mathbb{E}(\widetilde{\Gamma_3}; F_1 {\supset} F_2)) \vdash H$} 
    \RightLabel{$\supset L$} 
    \BinaryInfC{$D_6: \Gamma_2(\mathbb{E}(\widetilde{\Gamma_3}; F_1 {\supset} F_2)) \vdash H$} 
    \RightLabel{$\Cut$} 
    \BinaryInfC{$\Gamma_2(\mathbb{E}(\widetilde{\Gamma_3}; \Gamma_1(H_1 \vee H_2)))
    \vdash H$}  
    \DisplayProof 
    }
  \end{center}
   where {\small $D_1: \Gamma_1(H_1) \vdash F_1 {\supset} F_2$} 
  and {\small $D_2: \Gamma_1(H_2) \vdash F_1 {\supset} F_2$}.
  The cut formula {\small $F_1 {\supset} F_2$} 
  is not the principal on the left premise. In this case,  
  we simply apply {\Cut} on the pairs: {\small ($D_1, D_6$)} and 
  {\small ($D_2, D_6$)}, to conclude:{\ }\\  
  {\small 
  \begin{center}
    \AxiomC{$D_1$} 
    \AxiomC{$D_6$} 
    \RightLabel{$\Cut$} 
    \BinaryInfC{$\Gamma_2(\mathbb{E}(\widetilde{\Gamma_3}; \Gamma_1(H_1))) \vdash H$}
    \AxiomC{$D_2$} 
    \AxiomC{$D_6$} 
    \RightLabel{$\Cut$} 
    \BinaryInfC{$\Gamma_2(\mathbb{E}(\widetilde{\Gamma_3}; \Gamma_1(H_2))) 
    \vdash H$} 
    \RightLabel{$\vee L$} 
    \BinaryInfC{$\Gamma_2(\mathbb{E}(\widetilde{\Gamma_3}; \Gamma_1(H_1 \vee H_2)))
    \vdash H$} 
    \DisplayProof
  \end{center}
  } 
  {\ }\\
  Of course, for this particular permutation to be correct, we must be 
  able to demonstrate, in the permuted derivation tree, 
  that {\small $\mathbb{E}(\widetilde{\Gamma_3}; 
  \Gamma_1(H_1 \vee H_2)) = 
  \mathbb{E}'(\widetilde{\Gamma_3}) \star \Gamma_1(H_1 \vee H_2)$} 
  with {\small $\star$} either a semi-colon or a comma, that 
  {\small $\mathbb{E}(\widetilde{\Gamma_3}; \Gamma_1(H_1)) = 
      \mathbb{E}'(\widetilde{\Gamma_3}) \star \Gamma_1(H_1)$}, and that  
  {\small $\mathbb{E}(\widetilde{\Gamma_3}; \Gamma_1(H_2)) = 
      \mathbb{E}'(\widetilde{\Gamma_3}) \star \Gamma_1(H_2)$}. But this is 
  vacuous since the cut formula which is replaced with 
  the structure {\small $\Gamma_1(H_1)$} or {\small $\Gamma_1(H_2)$} 
  is a formula. \\
  \indent Cases that remain are those for which  
  both premises of the cut instance have the cut formula 
  as the principal. We go through each to conclude the proof. 
  \begin{description}
    \item[($\wedge L, \wedge R$):]{\ }\\ 
      {\small 
      \begin{center}
	 \AxiomC{$D_1: \Gamma_1 \vdash F_1$} 
	 \AxiomC{$D_2: \Gamma_1 \vdash F_2$} 
	 \RightLabel{$\wedge R$} 
	 \BinaryInfC{$\Gamma_1 \vdash F_1 \wedge F_2$} 
	 \AxiomC{$D_3: \Gamma_2(F_1 ; F_2) \vdash H$} 
	 \RightLabel{$\wedge L$} 
	 \UnaryInfC{$\Gamma_2(F_1 \wedge F_2) \vdash H$} 
	 \RightLabel{$\Cut$} 
	 \BinaryInfC{$\Gamma_2(\Gamma_1) \vdash H$} 
	 \DisplayProof
      \end{center}
      }  
      $\Rightarrow${\ }\\
      {\small 
      \begin{center} 
	\AxiomC{$D_2$} 
	 \AxiomC{$D_1$} 
	 \AxiomC{$D_3$} 
	 \RightLabel{$\Cut$} 
	 \BinaryInfC{$\Gamma_2(\Gamma_1; F_2) \vdash H$}  
	 \RightLabel{$\Cut_{CS}$} 
	 \BinaryInfC{$\Gamma_2(\Gamma_1) \vdash H$} 
	 \DisplayProof
      \end{center} 
      } 
    \item[($\vee L, \vee R$): ]{\ }\\ 
      {\small
      \begin{center}
	\AxiomC{$D_1: \Gamma_1 \vdash F_i \quad (i \in \{1,2\})$} 
	 \RightLabel{$\vee R$} 
	 \UnaryInfC{$\Gamma_1 \vdash F_1 \vee F_2$} 
	 \AxiomC{$D_2: \Gamma_2(F_1) \vdash H$} 
	 \AxiomC{$D_3: \Gamma_2(F_2) \vdash H$} 
	 \RightLabel{$\vee L$} 
	 \BinaryInfC{$\Gamma_2(F_1 \vee F_2) \vdash H$} 
	 \RightLabel{$\Cut$} 
	 \BinaryInfC{$\Gamma_2(\Gamma_1) \vdash H$} 
	 \DisplayProof
       \end{center} }
      $\Rightarrow$ 
      {\ }\\  
      {\small 
      \begin{center}
	\AxiomC{$D_1$} 
	\AxiomC{$D_{(2\: or\: 3)}$} 
	\RightLabel{$\Cut$} 
	\BinaryInfC{$\Gamma_2(\Gamma_1) \vdash H$} 
	\DisplayProof
      \end{center} 
      } 
      Whether {\small $D_2$} or {\small $D_3$} for the 
      right premise sequent depends on the 
      value of {\small $i$}.
    \item[($\supset L, \supset R$): ]{\ }\\ 
      \begin{center} 
	\scalebox{0.83}{
	\hspace{-0.8cm}
	\AxiomC{$D_1: \Gamma_3; F_1  \vdash F_2$} 
	\RightLabel{$\supset R$} 
	\UnaryInfC{$D_4: \Gamma_3 \vdash F_1 {\supset} F_2$} 
    \AxiomC{$D_2: \mathbb{E}(\widetilde{\Gamma_1}; F_1 {\supset} F_2) \vdash F_1$} 
    \AxiomC{$D_3: \Gamma_2(F_2; \mathbb{E}(\widetilde{\Gamma_1}; F_1 {\supset} F_2)) \vdash H$} 
	\RightLabel{$\supset L$} 
    \BinaryInfC{$\Gamma_2(\mathbb{E}(\widetilde{\Gamma_1}; F_1 {\supset} F_2))
	\vdash H$}   
	\RightLabel{$\Cut$} 
    \BinaryInfC{$\Gamma_2(\mathbb{E}(\widetilde{\Gamma_1}; \Gamma_3)) 
	\vdash H$} 
	\DisplayProof 
	}
      \end{center}
      $\Rightarrow${\ }\\  
      {\small 
      \begin{center} 
	 \AxiomC{$D_4$} 
	 \AxiomC{$D_2$} 
	 \RightLabel{$\Cut$} 
     \BinaryInfC{$\mathbb{E}(\widetilde{\Gamma_1}; \Gamma_3) \vdash F_1$}  
	 \AxiomC{$D_1$}
	 \RightLabel{$\Cut$} 
     \BinaryInfC{$\Gamma_3; \mathbb{E}(\widetilde{\Gamma_1}; \Gamma_3) \vdash 
	 F_2$}   
	 \AxiomC{$D_4$} 
	 \AxiomC{$D_3$} 
	 \RightLabel{$\Cut$} 
     \BinaryInfC{$\Gamma_2(F_2; \mathbb{E}(\widetilde{\Gamma_1}; \Gamma_3)) \vdash
	 H$}
	 \RightLabel{$\Cut_{CS}$} 
	 \BinaryInfC{$\Gamma_2(\Gamma_3; 
         \mathbb{E}(\widetilde{\Gamma_1}; \Gamma_3)) \vdash H$}  
	 \RightLabel{Proposition \ref{admissible_weakening_LBI3}} 
	 \dottedLine  
     \UnaryInfC{$\Gamma_2(\widetilde{\Gamma_1}; \Gamma_3;
         \mathbb{E}(\widetilde{\Gamma_1}; \Gamma_3)) \vdash H$}
	 \RightLabel{Proposition \ref{admissible_EA2}} 
	 \dottedLine 
     \UnaryInfC{$\Gamma_2(\mathbb{E}(\widetilde{\Gamma_1}; \Gamma_3); 
         \mathbb{E}(\widetilde{\Gamma_1}; \Gamma_3)) \vdash 
	 H$}  
	 \dottedLine
	 \RightLabel{Proposition \ref{admissible_contraction_LBI3}} 
     \UnaryInfC{$\Gamma_2(\mathbb{E}(\widetilde{\Gamma_1}; 
     \Gamma_3)) \vdash 
	 H$} 
	 \DisplayProof
      \end{center}
      } 
      \noindent where a dotted line denotes that the derivation step
      is depth-preserving.  
    \item[($* L, * R$): ]{\ } 
      {\small 
      \begin{center}
	\AxiomC{$D_1: Re_i \vdash F_1$} 
	\AxiomC{$D_2: Re_j \vdash F_2$} 
	\RightLabel{$* R$} 
	\BinaryInfC{$\Gamma_1 \vdash F_1 * F_2$} 
	\AxiomC{$D_3: \Gamma_2(F_1, F_2) \vdash H$}  
	\RightLabel{$* L$} 
	\UnaryInfC{$\Gamma_2(F_1 * F_2) \vdash H$}
	\RightLabel{$\Cut$} 
	\BinaryInfC{$\Gamma_2(\Gamma_1) \vdash H$} 
	\DisplayProof
      \end{center}
      } 
      $\Rightarrow${\ }\\  
      {\small 
      \begin{center} 
	\AxiomC{$D_2$}
       \AxiomC{$D_1$} 
       \AxiomC{$D_3$} 
       \RightLabel{$\Cut$}
       \BinaryInfC{$\Gamma_2(Re_i, F_2) \vdash H$} 
       \RightLabel{$\Cut$} 
       \BinaryInfC{$\Gamma_2(Re_i, Re_j) \vdash H$} 
       \RightLabel{Proposition \ref{admissible_weakening_LBI3}} 
       \dottedLine 
       \UnaryInfC{$\Gamma_2(\Gamma_1) \vdash H$} 
       \DisplayProof
      \end{center}
      }  
    \item[(${\text{\wand}} L, {\text{\wand}} R$): ]{\ } 
      \begin{center} 
	\scalebox{0.84}{ 
	\hspace{-1cm}
	\AxiomC{$D_1: \Gamma_1, F_1 \vdash F_2$} 
	\RightLabel{$\text{\wand} R$} 
	\UnaryInfC{$D_4: \Gamma_1 \vdash F_1 {\text{\wand}} F_2$} 
	\AxiomC{$D_2: Re_i \vdash F_1$} 
    \AxiomC{$D_3: \Gamma_2((\widetilde{Re_j}, F_2); (\widetilde{\Gamma'}, 
        \mathbb{E}(\widetilde{\Gamma_3}; F_1 {\text{\wand}} F_2))) \vdash H$} 
	\RightLabel{$\text{\wand} L_1$} 
    \BinaryInfC{$\Gamma_2(\widetilde{\Gamma'}, \mathbb{E}(\widetilde{\Gamma_3}; F_1 {\text{\wand}} F_2)) \vdash H$} 
	\RightLabel{$\Cut$} 
    \BinaryInfC{$\Gamma_2(\widetilde{\Gamma'}, 
        \mathbb{E}(\widetilde{\Gamma_3}; \Gamma_1)) \vdash H$} 
	\DisplayProof 
	}
      \end{center} 
      $\Rightarrow${\ }\\  
      {\small 
      \begin{center} 
	\AxiomC{$D_2$}
	\AxiomC{$D_1$} 
	\AxiomC{$D_4$} 
	\AxiomC{$D_3$} 
	\RightLabel{$\Cut$} 
    \BinaryInfC{$\Gamma_2((\widetilde{Re_j}, F_2); (\widetilde{\Gamma'}, 
        \mathbb{E}(\widetilde{\Gamma_3}; \Gamma_1))) \vdash H$}   
	\RightLabel{$\Cut$} 
    \BinaryInfC{$\Gamma_2((\widetilde{Re_j}, \Gamma_1, F_1); 
        (\Gamma', \mathbb{E}(\widetilde{\Gamma_3}; \Gamma_1))) \vdash H$} 
	\RightLabel{$\Cut$} 
    \BinaryInfC{$\Gamma_2((\widetilde{Re_j}, \Gamma_1, Re_i); 
        (\widetilde{\Gamma'}, \mathbb{E}(\widetilde{\Gamma_3}; \Gamma_1))) \vdash H$}  
	\RightLabel{Proposition \ref{admissible_weakening_LBI3}} 
	\dottedLine 
    \UnaryInfC{$\Gamma_2((\widetilde{\Gamma'}, (\widetilde{\Gamma_3}; \Gamma_1)); 
        (\widetilde{\Gamma'}, \mathbb{E}(\widetilde{\Gamma_3}; \Gamma_1))) \vdash H$} 
	\RightLabel{Proposition \ref{admissible_EA2}} 
	\dottedLine 
    \UnaryInfC{$\Gamma_2((\widetilde{\Gamma'}, 
        \mathbb{E}(\widetilde{\Gamma_3}; \Gamma_1)); (\widetilde{\Gamma'}, 
        \mathbb{E}(\widetilde{\Gamma_3}; \Gamma_1))) \vdash H$} 
	\RightLabel{Theorem \ref{admissible_contraction_LBI3}} 
	\dottedLine 
    \UnaryInfC{$\Gamma_2(\widetilde{\Gamma'}, 
        \mathbb{E}(\widetilde{\Gamma_3}; \Gamma_1)) \vdash H$} 
	\DisplayProof
      \end{center}
      } 
  \end{description} 
\end{document}